\documentclass[review]{elsarticle}

\usepackage{lineno,hyperref}
\usepackage{graphicx}
\usepackage{enumerate}
\usepackage[small,bf,hang]{caption} 
\usepackage{subcaption}
\usepackage{amsmath}
\usepackage{amssymb}
\usepackage{rotating}
\usepackage{xcolor}		 	
\usepackage{verbatim}		
\usepackage{colortbl}
\usepackage{booktabs}
\usepackage{dcolumn}
\usepackage{expdlist}  
\usepackage{float}
\usepackage{url}
\usepackage{fullpage}
\usepackage{mathtools}
\usepackage{amsthm}
\usepackage{cleveref}[2012/02/15]
\usepackage{pifont}

\usepackage{longtable}
\usepackage{threeparttable}
\usepackage{threeparttablex}
\usepackage[ruled,vlined]{algorithm2e}
\usepackage{optidef}

\usepackage{mathptmx}   
\usepackage{pdfpages}
\usepackage{adjustbox}
\usepackage{array}
\usepackage[graphicx]{realboxes}
\usepackage{makecell}
\usepackage[onehalfspacing]{setspace}
\usepackage{cclicenses}

\newcolumntype{R}[2]{%
    >{\adjustbox{angle=#1,lap=\width-(#2)}\bgroup}%
    l%
    <{\egroup}%
}
   
\DeclarePairedDelimiterX\setc[2]{\{}{\}}{\,#1 \;\delimsize\vert\; #2\,}

\DeclareMathOperator*{\argmin}{arg\,min}

\newtheorem{lemma}{Lemma}
\newtheorem{definition}{Definition}
\newtheorem{theorem}{Theorem}
\newtheorem{corollary}{Corollary}

\newcommand{\ssum}[1]{\sum_{\substack{#1}}}
\newcommand{\smax}[1]{\max_{\substack{#1}}}
\newcommand{\smin}[1]{\min_{\substack{#1}}}

\newcommand{\cmark}{\ding{51}}%
\newcommand{\xmark}{\ding{55}}%

\SetKwInput{KwData}{Input}
\SetKwInput{KwResult}{Output}

\SetKwFunction{Range}{range}

\SetKwProg{FnA}{Function}{}{end}\SetKwFunction{FRecursA}{calculateNumberIMSWeight}%
\SetKwProg{FnB}{Function}{}{end}\SetKwFunction{FRecursB}{updateNumberSubsetsWeight}%
\SetKwProg{FnC}{Function}{}{end}\SetKwFunction{FRecursC}{estimateHCN}%
\SetAlgoLongEnd

\SetKwFunction{Range}{range}
\SetKwFor{While}{while}{:}{fintq}%
\SetKw{KwTo}{to}\SetKwFor{For}{for}{\string:}{}%
\newcommand{\forcondA}{$i=1$  \KwTo{$n$}}
\newcommand{\forcondB}{$k=c+1-w_i$  \KwTo{$c$}}
\newcommand{\forcondC}{$k=c$  \KwTo{$w[i-1]$}}
\crefformat{footnote}{#2\footnotemark[#1]#3}

\journal{XXX}

\biboptions{authoryear}

\begin{document}

\begin{frontmatter}

\title{Features for the 0-1 knapsack problem based on inclusionwise maximal solutions}


\author[KULAK]{Jorik Jooken\corref{mycorrespondingauthor}}
\ead{jorik.jooken@kuleuven.be}
\cortext[mycorrespondingauthor]{Corresponding author}
\author[KULAK,UGENT,FLANDERSMAKE,UANTWERPEN]{Pieter Leyman}
\ead{pieter.leyman@ugent.be;pieter.leyman@kuleuven.be;pieter.leyman@uantwerpen.be}
\author[KULAK]{Patrick De Causmaecker}
\ead{patrick.decausmaecker@kuleuven.be}
\address[KULAK]{Department of Computer Science, KU Leuven Kulak, Etienne Sabbelaan 53, 8500 Kortrijk, Belgium}
\address[UGENT]{Department of Industrial Systems Engineering and Production Design, Technologiepark Zwijnaarde 46, 9052 Zwijnaarde, Belgium}
\address[FLANDERSMAKE]{Industrial Systems Engineering, Flanders Make@UGent, Sint-Martens-Latemlaan 2B, 8500 Kortrijk, Belgium}
\address[UANTWERPEN]{Department of Engineering Management, University of Antwerp, Prinsstraat 13, 2000 Antwerp, Belgium}




\begin{abstract}
Decades of research on the 0-1 knapsack problem led to very efficient algorithms that are able to quickly solve large problem instances to optimality. This prompted researchers to also investigate whether relatively small problem instances exist that are hard for existing solvers and investigate which features characterize their hardness. Previously the authors proposed a new class of hard 0-1 knapsack problem instances and demonstrated that the properties of so-called inclusionwise maximal solutions (IMSs) can be important hardness indicators for this class. In the current paper, we formulate several new computationally challenging problems related to the IMSs of arbitrary 0-1 knapsack problem instances. Based on generalizations of previous work and new structural results about IMSs, we formulate polynomial and pseudopolynomial time algorithms for solving these problems. From this we derive a set of 14 computationally expensive features, which we calculate for two large datasets on a supercomputer in approximately 540 CPU-hours. We show that the proposed features contain important information related to the empirical hardness of a problem instance that was missing in earlier features from the literature by training machine learning models that can accurately predict the empirical hardness of a wide variety of 0-1 knapsack problem instances. Using the instance space analysis methodology, we also show that hard 0-1 knapsack problem instances are clustered together around a relatively dense region of the instance space and several features behave differently in the easy and hard parts of the instance space.
\end{abstract}


\end{frontmatter}


	\section{Introduction}
\label{introductionSection}
The 0-1 knapsack problem is a fundamental NP-hard optimization problem \citep{Kellerer:2004a} in which we are given a knapsack with capacity $c$ and $n$ items with weights $w_1, w_2, \ldots, w_n$ and profits $p_1, p_2, \ldots, p_n$. The goal is to select a subset of these $n$ items in such a way that the total weight of this subset does not exceed $c$ and the total profit of this subset is maximized. More formally, the 0-1 knapsack problem can be described by the following integer program:

\begin{maxi!}|l|
	{}{& & \ \ \sum_{i=1}^{n}p_i x_i}{}{} \label{objectiveFunction}
	\addConstraint{& \ \ \sum_{i=1}^{n}w_i x_i}{\leq c}{} \label{capacityConstraint}
	\addConstraint{x_{i}}{\in \{0,1\},}{\ \ \forall i \in \{1,2,\ldots,n\}} \label{domainConstraint}
\end{maxi!}
Here, the $n$ binary variables $x_1, x_2, \ldots, x_n$ are the decision variables that determine whether an item should be put into the knapsack. The problem parameters $c$, $n$, $w_1, w_2, \ldots, w_n$ and $p_1, p_2, \ldots, p_n$ are given as part of the input. For the rest of this paper assume that all problem parameters are strictly positive integers, that the weights are ordered in non-increasing order (i.e. $w_1 \geq w_2 \geq \ldots \geq w_n$), that not all items can be put into the knapsack at the same time (i.e. $\sum_{i=1}^{n}w_i > c$) and that any item can be put into the knapsack individually (i.e. $w_i \leq c, \forall i: 1 \leq i \leq n$).

The 0-1 knapsack problem is an important optimization problem that is closely related to several other important optimization problems. For example, one can solve problem instances of the bounded knapsack problem, the unbounded knapsack problem and binary integer programming by solving an equivalent problem instance of the 0-1 knapsack problem \citep{Pisinger:2000,Andonov:2000,Kellerer:2004a}. The 0-1 knapsack problem also occurs as a column generation subproblem of the cutting stock problem \citep{Vanderbeck:1999} and it is a special case of a wide variety of problems (e.g. the travelling thief problem \citep{Bonyadi:2013}, the knapsack problem with conflicts \citep{Pferschy:2009,Coniglio:2021} and the multidimensional knapsack problem \citep{Freville:2004}).

Given the importance of the 0-1 knapsack problem, researchers have investigated this problem for several decades. This led to several very efficient algorithms, some of which are able to solve most large problem instances from the literature to optimality in a few seconds (e.g. the algorithms \textit{\textbf{MT1}} \citep{Martello:1977}, \textit{\textbf{MT2}} \citep{Martello:1988}, \textit{\textbf{Expknap}} \citep{Pisinger:1995}, \textit{\textbf{Minknap}} \citep{Pisinger:1997} and \textit{\textbf{Combo}} \citep{Martello:1999}). The practical success of these algorithms also led researchers to investigate several fundamental questions concerning the hardness of 0-1 knapsack problem instances. For instance, researchers are interested in whether hard 0-1 knapsack problem instances exist, how such instances could be generated and which features characterize their hardness. An answer to these questions would lead to deeper insights about the strengths and weaknesses of the algorithms and it could lead to new knowledge that can in turn be used to create more powerful algorithms. However, these questions are also challenging for several reasons and, despite decades of research, multiple questions concerning problem instance hardness remain unanswered. For instance, only a few features (e.g. the number of items or the knapsack capacity) are known to influence the hardness of a problem instance and it is not entirely clear which other features are important hardness indicators. Such features could in principle be the result of arbitrary functions that map a problem instance to a real number, but researchers are mostly interested in finding features that contain important knowledge related to the hardness of a problem instance. Another difficulty is that the concept of problem instance hardness would ideally be studied independent of existing implementations of algorithms such that one could draw general conclusions for all possible algorithms. However, in this setting it is very hard to make powerful statements without sacrificing generality. A practical solution that researchers have used to cope with this is to instead use so-called empirical hardness models \citep{Hutter:2014} in which one uses the running time of a state-of-the-art algorithm as a proxy for problem instance hardness. This approach was also followed in Pisinger's seminal paper entitled ``Where are the hard knapsack problems?" \citep{Pisinger:2005} and in two more recent papers \citep{Smith-Miles:2021,Jooken:2022} that revisit this question from a new angle. In the current paper, we also follow this approach.

In previous work \citep{Jooken:2022}, the authors explained that the properties of \textit{inclusionwise maximal solutions} (further abbreviated as IMSs) can have an influence on the hardness of a problem instance. \cite{Jooken:2022} define an IMS as follows:

\begin{definition}
Consider a problem instance consisting of a knapsack with capacity $c$ and $n$ items with profits $p_{1}, p_{2}, \ldots, p_{n}$ and weights $w_{1}, w_{2}, \ldots, w_{n}$. A solution $x_{1}, x_{2}, \ldots, x_{n}$ is called inclusionwise maximal (relative to the given problem instance) if it is a feasible solution (i.e. $\sum_{i=1}^{n}w_{i} x_{i} \leq c$) and there does not exist an item $j$ that can be added to the knapsack without violating the capacity constraint (i.e. $\nexists j: x_{j}=0 \land w_{j}+\sum_{i=1}^{n}w_{i} x_{i} \leq c$).
\end{definition}

In particular, \cite{Jooken:2022} showed the existence of a class of problem instances (so-called noisy multi-group exponential (NMGE) problem instances) for which the objective function value of any IMS is nearly equal to the optimal one. Because of this, such solutions are hard to prune from the search space using upper and lower bounds and this can make NMGE problem instances hard to solve to optimality.

The theoretical results obtained by \cite{Jooken:2022} help us to obtain deeper insights into the structure of the search space associated with NMGE problem instances. However, the earlier work also leaves two important gaps that we aim to address in the current paper. Firstly, the obtained results are all related to the objective function value of an IMS, but several other important features such as the number of IMSs or the distribution of the total weight of IMSs have not been studied before. One of the most important reasons for this is that calculating these features is hard (sometimes even harder than solving the problem instance to optimality) and previously it was not clear how to calculate such features with an algorithm that is not exponential in terms of $n$. In the current paper, however, we show that these features can be calculated in polynomial or pseudopolynomial time. Secondly, the results obtained by \cite{Jooken:2022} use the specific structure of NMGE problem instances and the results are only applicable to these problem instances. We aim to generalize these results to arbitrary 0-1 knapsack problem instances.

The main contributions of the current paper are as follows:
\begin{enumerate}
\item We formulate several new computationally challenging problems related to the IMSs of a problem instance and we prove structural results of IMSs that allow us to solve these problems in polynomial or pseudopolynomial time. We derive several features from the solutions of these problems and from subexpressions that play a crucial role in the results that we prove.
\item We generalize several theoretical results obtained by \cite{Jooken:2022} to arbitrary 0-1 knapsack problem instances and show how features can be derived from these results.
\item All algorithms to calculate the features were implemented. We calculated these features, as well as a set of existing features, for two large sets of problem instances from the literature in a large-scale experiment using a supercomputer for approximately 540 CPU-hours. We make these features publicly available at \url{https://github.com/JorikJooken/knapsackFeatures} to help other researchers who want to build upon this work without having to face this computational challenge again.
\item We empirically show that the features from the current paper are important hardness indicators by training several machine learning models to predict the running time of the state-of-the-art knapsack algorithm \textit{\textbf{Combo}} and we show that the most accurate predictions are obtained when our features are used as input for the models. We also show that these features contain important information that was not present in earlier features from the literature. We determine important subsets of features and perform an instance space analysis using this subset. This allows us to visualize the features in a two-dimensional space and connect feature patterns in the instance space with problem instance hardness.
\end{enumerate}

The rest of this paper is structured as follows: in Section \ref{relatedWork} we give an overview of related work that is important to understand the broader context and put the current paper into perspective. In Section \ref{featureSection} we prove several results related to the IMSs of a problem instance, we derive features from these results and show how these features can be computed in polynomial or pseudopolynomial time. In Section \ref{experimentsSection} we discuss several experiments that we conducted based on these features and show that these features are important hardness indicators. Finally, we state the most important conclusions and discuss opportunities for further work in Section \ref{conclusionsSection}.

	\section{Related work}
\label{relatedWork}
In this section we give an overview of important related work from the literature that is useful to provide a context for the results presented in the current paper. The literature on the 0-1 knapsack problem is very rich and we do not attempt to give an exhaustive overview, but rather focus on the work that is relevant for the current paper. For more complete overviews, the interested reader is referred to the excellent overviews in \cite{Martello:1990}, \cite{Pisinger:1998}, \cite{Kellerer:2004a} and \cite{Cacchiani:2022}.

Most successful algorithms for the 0-1 knapsack problem are based on the principles of branch-and-bound, dynamic programming or hybrid versions in which unpromising dynamic programming states are fathomed based on bounds. A series of gradual improvements led to the development of the algorithms \textit{\textbf{MT1}} \citep{Martello:1977}, \textit{\textbf{MT2}} \citep{Martello:1988}, \textit{\textbf{Expknap}} \citep{Pisinger:1995}, \textit{\textbf{Minknap}} \citep{Pisinger:1997} and \textit{\textbf{Combo}} \citep{Martello:1999}. The most important of these algorithms is without a doubt the \textit{\textbf{Combo}} algorithm. It is an exact algorithm that is able to solve most problem instances from the literature in a couple of seconds and it still represents the currently best known algorithm, despite being published more than twenty years ago. It combines several algorithmic ideas and is reasonably similar to the \textit{\textbf{Minknap}} algorithm, but it introduces new techniques when the number of dynamic programming states exceeds a certain threshold. An important idea that \textit{\textbf{Combo}} uses, consists of adding valid inequalities (cardinality constraints) to the integer programming formulation which are then surrogate relaxed to obtain tighter dual bounds. Previous research \citep{Martello:1999,Pisinger:2005,Smith-Miles:2021,Jooken:2022} has shown that \textit{\textbf{Combo}} is almost always the fastest algorithm amongst these five algorithms. In the current paper we use an empirical hardness model in which the running time of \textit{\textbf{Combo}} is used as a proxy for problem instance hardness.

Since \textit{\textbf{Combo}} is able to solve most large problem instances from the literature in a couple of seconds and the knapsack problem is (weakly) NP-hard, researchers have also shown interest in finding problem instances that are able to pose a bigger challenge for knapsack algorithms. This research is challenging, since most researchers expect that such hard problem instances should indeed exist, but it is unclear how these problem instances can be found and there is only a limited insight into which features play an important role in hard knapsack problem instances. The research concerned with 0-1 knapsack problem instances can broadly be categorized into two groups of papers. 

A first group of papers studies knapsack problem instances for which the coefficients can be very large (e.g. exponentially large in terms of the number of items $n$). The hardness of these problem instances is not practically tested by using existing algorithm implementations like \textit{\textbf{Combo}} (most implementations only support 32-bit or 64-bit integers), but rather tested against hypothetical sets of algorithms (e.g. using different assumptions about the strengths of hypothetical bounds that these algorithms can use). Amongst this first group of papers, we mention \cite{Chvatal:1980}, \cite{Gu:1999} and \cite{Jukna:2011} which describe hard problem instances with very large coefficients for different sets of hypothetical algorithms. 

A second group of papers considers more practical problem instances for which the coefficients can be handled by existing algorithm implementations and the running time is investigated. For this second group of papers, we mention Pisinger's seminal paper entitled ``Where are the hard knapsack problems?" \citep{Pisinger:2005} in which 13 different classes of problem instances are studied and the running time of \textit{\textbf{Combo}} is investigated. In a more recent paper by the authors \citep{Jooken:2022}, a new class of hard problem instances was proposed and several theorems were proven which suggest that the properties of inclusionwise maximal solutions can be important hardness indicators for this new class of problem instances. These problem instances take much longer to solve than previous problem instances from the literature, despite being smaller, and \textit{\textbf{Combo}} was the only algorithm for which it was shown that many of these problem instances can be solved in less than 2 hours. Another recent paper \citep{Smith-Miles:2021} revisits Pisinger's question (``Where are the hard knapsack problems?") using the Instance Space Analysis methodology \citep{Smith-Miles:2012,Smith-Miles:2014,Smith-Miles:2015}. This methodology studies problem instances based on a set of features that (usually) contain important information regarding the running time of one or more algorithms. The problem instances are regarded as points in a high-dimensional space and they are projected to a two-dimensional space such that certain desirable properties are met (e.g. the PILOT method of \cite{Munoz:2018} attempts to find a projection that results in a linear distribution of the features and the running times). This allows one to see both the features and the running times in the same space such that the relationship between feature patterns and problem instance hardness can be visually observed. Features are clearly of crucial importance to this methodology, which allows one to investigate a wide range of interesting questions such as  ``Where are the hard problem instances located?'', ``Which features correlate with problem instance hardness?'', ``How similar are two given problem instances?'' and ``Is a given dataset of problem instances varied enough to be a representative benchmark with regard to a given set of features?''.

Features (and the relationship with problem instance hardness) also play an important role in other related studies. For instance, for several combinatorial optimization problems a remarkable phenomenon occurs when the value of a specifically chosen feature enters or exits a certain interval. Around this interval, the difficulty of a problem instance (based on an empirical hardness model) drastically changes from very easy to very hard. These are the so-called phase transitions in combinatorial optimization problems \citep{Hartmann:2006} and it has been conjectured that such a crucial feature and critical region for that feature exist for most problems \citep{Cheeseman:1991,Achlioptas:2005}. Clear examples of phase transitions have been found for the SAT problem \citep{Mitchell:1992} and the travelling salesman problem \citep{Gent:1996,Smith-Miles:2010}. Features also play an important role in the context of algorithm selection \citep{Rice:1976,Hall:2007,Kerschke:2019}. In algorithm selection, the goal is to select an algorithm from a (large) set of algorithms that can solve a set of problem instances as fast (or as well) as possible without having to run all algorithms from the given set. Instead, an algorithm selector is trained by using a machine learning model which learns to map problem instance features to a suitable algorithm. In the context of algorithm selection, it is also important how the cost of computing a certain feature compares with the cost of running the given algorithms (unlike the previous examples, where the informativeness of the features is the most important aspect). Since most features from the current paper are computationally expensive, they are more suitable for the previous examples where informativeness is more important rather than algorithm selection where the computational cost is also important. However, note that in some contexts such expensive features can still be useful for algorithm selection, even when their computational cost exceeds the cost of solving a problem instance. For example, \cite{Hutter:2014} describe model-based algorithm configuration \citep{Hutter:2011} and complex empirical analyses based on performance predictions \citep{Hutter:2010,Hutter:2013} as example applications for such expensive features.

	\section{Features related to inclusionwise maximal solutions}
\label{featureSection}
In the rest of this section, we consider a fixed 0-1 knapsack problem instance with a knapsack capacity of $c$ and $n$ items with weights $w_1, w_2, \ldots, w_n$ and profits $p_1, p_2, \ldots, p_n$. For the readers' convenience, we summarize all features that will be defined in the rest of this section in Table \ref{featuresOverviewTable}.

\begin{table}[h!] \scriptsize\centering 
	\begin{threeparttable}
		\caption{An overview of the features from the current paper.}
		\label{featuresOverviewTable} 
		\begin{tabular}{cccccccccccccccccccccc} \\
			\hline
			\noalign{\smallskip} 
			Feature & Domain & Obtained from & Algorithm implemented & Time class\\ 
			\noalign{\smallskip}
			\hline
			\noalign{\smallskip}
			\multicolumn{1}{l}{1. $\left|X\right|$} & $[1,2^n]$ & Subsection \ref{countingBasedFeatures}+\ref{dynamicProgramming} & Dynamic programming: $O(nc)$ & Pseudopolynomial\\ 
			\multicolumn{1}{l}{2. $\smin{\mathbf{x} \in X} \text{totalWeight}(\mathbf{x})$} & $[1,c]$ & Subsection \ref{countingBasedFeatures}+\ref{dynamicProgramming} & Dynamic programming: $O(nc)$ & Pseudopolynomial\\ 
			\multicolumn{1}{l}{3. $\smax{\mathbf{x} \in X} \text{totalWeight}(\mathbf{x})$} & $[1,c]$ & Subsection \ref{countingBasedFeatures}+\ref{dynamicProgramming} & Dynamic programming: $O(nc)$ & Pseudopolynomial\\ 
			\multicolumn{1}{l}{4. $\overline{\text{totalWeight}}$} & $[1,c]$ & Subsection \ref{countingBasedFeatures}+\ref{dynamicProgramming} & Dynamic programming: $O(nc)$ & Pseudopolynomial\\ 
			\multicolumn{1}{l}{5. $\sigma^2$} & $[0,c^2]$  & Subsection \ref{countingBasedFeatures}+\ref{dynamicProgramming} & Dynamic programming: $O(nc)$ & Pseudopolynomial\\ 	
			\multicolumn{1}{l}{6. $t_1$} & $]0,\infty[$  & Running time (in seconds) of the algorithm & Dynamic programming: $O(nc)$ & Pseudopolynomial\\ 	
			\multicolumn{1}{l}{7. $b$} & $[1,n]$ & Theorem \ref{lowerBoundTheorem} & For-loop: $O(n)$ & Polynomial\\
			\multicolumn{1}{l}{8. $f$} & $[1,n]$ & Theorem \ref{lowerBoundTheorem} & Sorting: $O(n \log(n))$ & Polynomial\\ 	 	
			\multicolumn{1}{l}{9. $\sum_{i=1}^{f-1}(p_{\pi_i})+\frac{c-w_b+1-\sum_{i=1}^{f-1}w_{\pi_i}}{w_{\pi_f}} p_{\pi_f}$} & $[0, \sum_{i=1}^{n} p_i]$ & Theorem \ref{lowerBoundTheorem} & Sorting: $O(n \log(n))$ &Polynomial\\ 	 
			\multicolumn{1}{l}{10. $g^*$} & $[2,n-1]$ & Subsection \ref{boundsFeatures} & \cite{Gronlund:2017}: $O(ng^*)$ & Polynomial\\
			\multicolumn{1}{l}{11. $s_{g^*}$} & $[1,n+1-g^*]$ & Theorem \ref{cardinalityTheorem} & \cite{Gronlund:2017}: $O(ng^*)$ & Polynomial\\
			\multicolumn{1}{l}{12. $t_2$} & $]0,\infty[$  & Running time (in seconds) of the algorithm & \cite{Gronlund:2017}: $O(ng^*)$ & Polynomial\\ 	  
			\multicolumn{1}{l}{13. $z$} & $[0,s_{g^*}]$ & Theorem \ref{cardinalityTheorem} & \cite{Pisinger:2000}: $O(g^*c\log(c))$ & Pseudopolynomial\\
			\multicolumn{1}{l}{14. $t_3$} & $]0,\infty[$  & Running time (in seconds) of the algorithm & \cite{Pisinger:2000}: $O(g^*c\log(c))$ & Pseudopolynomial\\ 	  	
			\noalign{\smallskip}
			\hline
		\end{tabular}
	\end{threeparttable}
\end{table}

\subsection{Counting-based features}
\label{countingBasedFeatures}
Let $X$ be the set of all IMSs for a given problem instance, let $\mathbf{x}=x_{1}, x_{2}, \ldots, x_{n}$ be an IMS from the set $X$ and let $\text{totalWeight}(\mathbf{x})=\ssum{i=1}^{n} w_{i} x_{i}$. We are interested in calculating the following features related to the IMSs:
$$\left|X\right|$$
$$\smin{\mathbf{x} \in X} \text{totalWeight}(\mathbf{x})$$
$$\smax{\mathbf{x} \in X} \text{totalWeight}(\mathbf{x})$$
$$\overline{\text{totalWeight}}=\frac{\ssum{\mathbf{x} \in X} \text{totalWeight}(\mathbf{x})}{\left|X\right|}$$
$$\sigma^2=\frac{\ssum{\mathbf{x} \in X} (\text{totalWeight}(\mathbf{x})-\overline{\text{totalWeight}})^2}{\left|X\right|}$$
These features are all related to the number of IMSs and the distribution of their total weights (or equivalently their total unused capacity, which is equal to $c$ minus the total weight). Note that for typical problem instances the number of IMSs can be extremely large (e.g. for many problem instances from the experiments section we have $\left|X\right|>10^{100}$), which makes computing the above features very challenging. 

We now define $\text{numberIMSWeight}(k)$ ($\forall k: 1 \leq k \leq c$) as follows:
$$\text{numberIMSWeight}(k)=\left| \setc{\mathbf{x} \in X}{\text{totalWeight}(\mathbf{x})=k} \right|$$
Using this definition, all features can be rewritten in terms of $\text{numberIMSWeight}(k)$:
$$\left|X\right| = \sum_{k=1}^c \text{numberIMSWeight}(k)$$
$$\smin{\mathbf{x} \in X} \text{totalWeight}(\mathbf{x}) = \smin{1 \leq k \leq c} \setc{k}{\text{numberIMSWeight}(k)>0}$$
$$\smax{\mathbf{x} \in X} \text{totalWeight}(\mathbf{x})= \smax{1 \leq k \leq c} \setc{k}{\text{numberIMSWeight}(k)>0}$$
$$\overline{\text{totalWeight}}=\frac{\ssum{\mathbf{x} \in X} \text{totalWeight}(\mathbf{x})}{\left|X\right|}=\frac{\sum_{k=1}^c \text{numberIMSWeight}(k) \cdot k}{\left|X\right|}$$
$$\sigma^2=\frac{\ssum{\mathbf{x} \in X} (\text{totalWeight}(\mathbf{x})-\overline{\text{totalWeight}})^2}{\left|X\right|}=\frac{\sum_{k=1}^c (k-\overline{\text{totalWeight}})^2  \cdot \text{numberIMSWeight}(k)}{\left|X\right|}$$
Hence, if we are given $\text{numberIMSWeight}(k)$ ($\forall k: 1 \leq k \leq c$) we can calculate all features in $O(c)$ time. In what follows, we will derive a dynamic programming algorithm based on a suitable way to partition $X$ that allows us to compute $\text{numberIMSWeight}(k)$ ($\forall k: 1 \leq k \leq c$) with a time complexity of $O(nc)$ and a space complexity of $O(c)$.

\subsection{Dynamic programming algorithm}
\label{dynamicProgramming}
The set of all IMSs $X$ for a given problem instance can be partitioned as follows:
$$X=\text{Exclude}_{1} \cup \text{Exclude}_{2} \cup \ldots \cup \text{Exclude}_{n}$$
Here, $\text{Exclude}_{i}$ represents the set of all IMSs for the given problem instance such that the item with index $i$ is the least heavy unselected item. Recall from Section \ref{introductionSection} that the items are sorted in non-increasing order of weight ($w_j \geq w_{j+1}$) and thus $\text{Exclude}_{i}$ can be written as follows:
$$\text{Exclude}_{i} = \setc{\mathbf{x} \in X}{x_{i}=0 \land \forall j \in \mathbb{N}: i+1 \leq j \leq n \Rightarrow x_{j}=1}$$
This leads to a useful characterization of IMSs shown in Lemma \ref{excludeLemma} and we will later be able to use this to compute $\text{numberIMSWeight}(k)$.

\begin{lemma}
\label{excludeLemma}
Let $S$ be the set of all $2^n$ solutions (either feasible or infeasible) for a given problem instance. Now the following three conditions are equivalent $\forall \mathbf{s} \in S, \forall i (1 \leq i \leq n)$:
\begin{gather*}
\mathbf{s} \in \text{Exclude}_{i}\\
\Leftrightarrow\\
c+1-w_{i} \leq \text{totalWeight}(\mathbf{s}) \leq c \land s_{i}=0 \land \forall j \in \mathbb{N}: i+1 \leq j \leq n \Rightarrow s_{j}=1\\
\Leftrightarrow\\
c+1-\ssum{j=i}^{n} w_{j} \leq \ssum{j=1}^{i-1} s_{j} w_{j} \leq c-\ssum{j=i+1}^{n} w_{j} \land s_{i}=0 \land \forall j \in \mathbb{N}: i+1 \leq j \leq n \Rightarrow s_{j}=1
\end{gather*}
\end{lemma}
\begin{proof}
We first show that the first condition implies the second condition, which in turn implies the third condition. We then show that the third condition implies the second condition, which in turn implies the first condition.

$( \Rightarrow )$: Since $\mathbf{s} \in \text{Exclude}_{i}$, we have $s_{i}=0 \land \forall j \in \mathbb{N}: i+1 \leq j \leq n \Rightarrow s_{j}=1$, because of the definition of $\text{Exclude}_{i}$. Furthermore, $\mathbf{s}$ is an IMS and thus it must be feasible. In other words, we have $\text{totalWeight}(\mathbf{s}) = \ssum{j=1}^{i-1} s_{j} w_{j} + 0 w_{i} +\ssum{j=i+1}^{n} w_{j} \leq c$ and thus $\ssum{j=1}^{i-1} s_{j} w_{j} \leq c-\ssum{j=i+1}^{n} w_{j}$. Finally, since $s_{i}=0$ and $\mathbf{s}$ is an IMS we have $\text{totalWeight}(\mathbf{s})+w_{i} > c$ and thus $c+1-w_{i} \leq \text{totalWeight}(\mathbf{s})$ and also $c+1-\ssum{j=i}^{n} w_{j} \leq \ssum{j=1}^{i} s_{j} w_{j}=\ssum{j=1}^{i-1} s_{j} w_{j}$.

$( \Leftarrow )$: Since $ \ssum{j=1}^{i-1} s_{j} w_{j} \leq c-\ssum{j=i+1}^{n} w_{j}$, we also have $\ssum{j=1}^{i-1} s_{j} w_{j} + 0 w_{i} +\ssum{j=i+1}^{n} w_{j} = \text{totalWeight}(\mathbf{s}) \leq c$ and thus $\mathbf{s}$ is a feasible solution. Since $ c+1-\ssum{j=i}^{n} w_{j} \leq \ssum{j=1}^{i-1} s_{j} w_{j}$, we also have $c+1-w_{i} \leq \text{totalWeight}(\mathbf{s})$ and thus item $i$ cannot be added into the knapsack without violating the capacity constraint. This in turn implies $\nexists j: s_{j}=0 \land w_{j}+\text{totalWeight}(\mathbf{s}) \leq c$, because the items are ordered in non-increasing order of weight. Hence, $\mathbf{s}$ is an IMS for which $s_{i}=0 \land \forall j \in \mathbb{N}: i+1 \leq j \leq n \Rightarrow s_{j}=1$ and we conclude that $\mathbf{s} \in \text{Exclude}_{i}$, as desired.
\end{proof}
Lemma \ref{excludeLemma} allows us to reduce the problem of computing $\text{numberIMSWeight}(k)$ to a more simple one. This is shown in Theorem \ref{equalityTheorem}.
\begin{theorem}
\label{equalityTheorem}
Define $\text{numberSubsetsWeight}(i,k)$ ($\forall i: 1 \leq i \leq n, \forall k \in \mathbb{Z}$) as follows:
$$\text{numberSubsetsWeight}(i,k)=\left| \setc{A \subseteq \{w_{1}, w_{2}, \ldots, w_{{i-1}}\}}{\ssum{w \in A} w=k} \right|$$
Now the following equality holds ($\forall k: 1 \leq k \leq c$):
\begin{align*}
	\text{numberIMSWeight}(k) & = \ssum{i=1}^{n} \begin{cases} 
             \text{numberSubsetsWeight}(i,k-\ssum{j=i+1}^{n} w_{j}) & \text{if } c+1-w_{i} \leq k \leq c\\
		 0 & \text{else}\\
               \end{cases}
\end{align*}
\end{theorem}
\begin{proof}
We can rewrite $\text{numberIMSWeight}(k)$ as follows:
$$\text{numberIMSWeight}(k)=\left| \setc{\mathbf{x} \in X}{\text{totalWeight}(\mathbf{x})=k} \right|=\ssum{i=1}^{n} \left| \setc{\mathbf{x} \in \text{Exclude}_{i}}{\text{totalWeight}(\mathbf{x})=k} \right|$$
The first equality is due to the definition of $\text{numberIMSWeight}(k)$ and the second equality is due to the fact that $X$ can be partitioned such that
$$X=\text{Exclude}_{1} \cup \text{Exclude}_{2} \cup \ldots \cup \text{Exclude}_{n}$$
Using Lemma \ref{excludeLemma}, we directly obtain ($\forall i: 1 \leq i \leq n, \forall k \in \mathbb{Z}$):
\begin{align*}
    \left| \setc{\mathbf{x} \in \text{Exclude}_{i}}{\text{totalWeight}(\mathbf{x})=k} \right| & = \begin{cases} 
             \text{numberSubsetsWeight}(i,k-\ssum{j=i+1}^{n} w_{j}) & \text{if } c+1-w_{i} \leq k \leq c\\
		 0 & \text{else}\\
               \end{cases}
\end{align*}
and thus
\begin{align*}
	\text{numberIMSWeight}(k) & = \ssum{i=1}^{n} \begin{cases} 
              \text{numberSubsetsWeight}(i,k-\ssum{j=i+1}^{n} w_{j}) & \text{if } c+1-w_{i} \leq k \leq c\\
		 0 & \text{else}\\
               \end{cases}
\end{align*}
\end{proof}

The problem of calculating $\text{numberSubsetsWeight}(i,k)$ is (a counting version of) the famous subset sum problem \citep{Bellman:1966, Pisinger:1999, Cormen:2009,Bringmann:2017, Koiliaris:2019}. In this problem, we are given a multiset of integers and a target value $k$ and we are asked to determine whether a subset of this multiset exists such that the sum of its numbers is equal to $k$ (or count the number of subsets in our case). The counting version of this problem can be solved by slightly modifying Bellman's recursion for the decision version of the problem \citep{Bellman:1966}. More specifically, we have the following base cases  ($\forall i: 1 \leq i \leq n, \forall k < 0$):
$$\text{numberSubsetsWeight}(i,k)=0$$
and ($\forall k \geq 0$):
\begin{align*}
    \text{numberSubsetsWeight}(1,k) & = \begin{cases} 
             1 & \text{if } k=0\\
             0 & \text{else}\\
               \end{cases}
\end{align*}
and the following recursive case ($\forall i: 2 \leq i \leq n, \forall k \geq 0$):
$$\text{numberSubsetsWeight}(i,k) =\text{numberSubsetsWeight}(i-1,k-w_{{i-1}})+\text{numberSubsetsWeight}(i-1,k)$$
The recursive case is obtained by realizing that for every subset of $\{w_{1}, w_{2}, \ldots, w_{{i-1}}\}$ there are two options for the weight $w_{i-1}$ (it can either be or not be an element of that subset). Using this recursion, we can compute $\text{numberSubsetsWeight}(i,k_1)$ ($\forall k_1: 1 \leq k_1 \leq c$) using an algorithm with time and space complexity $O(c)$ if we are given $\text{numberSubsetsWeight}(i-1,k_2)$ ($\forall k_2: 1 \leq k_2 \leq c$).

The above insights lead to the following algorithm (pseudocode shown in Algorithm \ref{numberIMSWeightAlgorithm}) for computing $\text{numberIMSWeight}(k)$ ($\forall k: 1 \leq k \leq c$). We first initialize two arrays with $c+1$ zeros to store the result of $\text{numberSubsetsWeight}(i,k)$ for a fixed value of $i$ ($\forall k: 0 \leq k \leq c$) and the result of $\text{numberIMSWeight}(k)$ ($\forall k: 1 \leq k \leq c$). Then the variable $i$ is iterated through the values from 1 to $n$ and in every iteration we update the array that stores the result of $\text{numberSubsetsWeight}(i,k)$ by using the base case and recursive case from the previous paragraph and we also update the array that stores the result of $\text{numberIMSWeight}(k)$ ($\forall k: 1 \leq k \leq c$) by using Theorem \ref{equalityTheorem}. The updates in every iteration take $O(c)$ time and space, which leads to a time complexity of $O(nc)$ and a space complexity of $O(c)$ for the whole algorithm.

\begin{algorithm}[h!]
  \FnB{\FRecursA{$n$, $c$, $w$}}{
    \KwData{\\$n$: an integer that represents the number of items\\$c$: an integer that represents the knapsack capacity\\$w$: an array of $n+1$ integers, where $w[0]$ contains an unused dummy value and $w[1], w[2], \ldots, w[n]$ represent the weights of the items in non-increasing order}
    \KwResult{\\numberIMSWeight: an array of $c+1$ integers such that $\text{numberIMSWeight}[k]$ represents the number of IMSs with weight $k$ ($\forall k: 1 \leq k \leq c$)}
    \tcc{Start of code}
    {numberIMSWeight $\gets zeros(c+1)$}\;
    {numberSubsetsWeight $\gets zeros(c+1)$}\;
    {suffixSum $\gets \ssum{j=1}^{n} w[j]$}\;
   \For{\forcondA}{
	 \tcc{In iteration $i$, suffixSum stores $\ssum{j=i+1}^{n} w[j]$}
	{suffixSum $\gets \text{suffixSum}-w[i]$}\;
	\tcc{In iteration $i$, numberSubsetsWeight$[k]$ stores the function numberSubsetsWeight$(i,k)$}
	\tcc{Update numberSubsetsWeight: base case}
	 \uIf{$i=1$}{
		{numberSubsetsWeight$[0] \gets 1$}\;
    	}
	\tcc{Update numberSubsetsWeight: recursive case}
	\uElse{
		\SetKw{KwTo}{down to}\SetKwFor{For}{for}{\string:}{}%
		\For{\forcondC}{
		{{$\text{numberSubsetsWeight}[k] \gets \text{numberSubsetsWeight}[k-w[i-1]]+\text{numberSubsetsWeight}[k]$}\;}
    		}
	}
	\tcc{Update numberIMSWeight}
	\For{\forcondB}{
      		\uIf{$k<\text{suffixSum}$}{
      			{continue}\;
    		}
		\uElse{
			{$\text{numberIMSWeight}[k] \gets \text{numberIMSWeight}[k]+\text{numberSubsetWeight}[k-\text{suffixSum}]$}\;
		}
    	}
    }
    {\Return numberIMSWeight}\;
  }
  \caption{Calculates $\text{numberIMSWeight}(k)$ ($\forall k: 1 \leq k \leq c$) with a time complexity of $O(nc)$ and a space complexity of $O(c)$.}
  \label{numberIMSWeightAlgorithm}
\end{algorithm}

\subsection{Bounding-based features: generalizations of \cite{Jooken:2022}}
\label{boundsFeatures}
In \cite{Jooken:2022} we proved several inequalities for NMGE problem instances based on several structural characterizations of NMGE problem instances. These inequalities are defined in terms of the parameters of the generator for NMGE problem instances and are not applicable to other problem instances generated by other generators. In this subsection we generalize these inequalities for arbitrary 0-1 knapsack problem instances, regardless of how they were generated, by uncovering several structural results of arbitrary 0-1 knapsack problem instances.

We first prove a lower bound for the objective function value of any IMS in Theorem \ref{lowerBoundTheorem}.
\begin{theorem}
\label{lowerBoundTheorem}
Let $b = \max\{e \in \{1,\ldots,n\}|\sum_{i=e}^n w_i > c\}$, let $\pi_1, \pi_2, \ldots, \pi_n$ be a permutation of the indices $\{1, 2, \ldots, n\}$ obtained by sorting the items such that $\frac{p_{\pi_i}}{w_{\pi_i}} \leq \frac{p_{\pi_{i+1}}}{w_{\pi_{i+1}}}$ ($\forall i \in \{1, 2, \ldots, n-1\}$) and let $f = \min\{e \in \{1,\ldots,n\}|\sum_{i=1}^e w_{\pi_i} \geq c-w_b+1\}$. The objective function value of any IMS $\mathbf{x}=x_1, x_2, \ldots, x_n$ is bounded from below as follows:
$$\sum_{i=1}^n p_i x_i \geq \sum_{i=1}^{f-1}(p_{\pi_i})+\frac{c-w_b+1-\sum_{i=1}^{f-1}w_{\pi_i}}{w_{\pi_f}} p_{\pi_f}$$
\end{theorem}
\begin{proof}
The smallest possible objective function value of any IMS is equal to the objective function value of the following integer program:
\begin{mini!}[2]
	{}{\sum_{i=1}^{n}p_i x_i}{}{} \label{objectiveFunction1}
	\addConstraint{\sum_{i=1}^{n}w_i x_i}{\leq c}{} \label{capacityConstraint1}
	\addConstraint{\sum_{i=1}^{n}(w_i x_i)+w_j (1-x_j)}{\geq (c+1)(1-x_j)}{\ \ \ \forall j \in \{1,2,\ldots,n\}} \label{IMSConstraints}
	\addConstraint{x_j \in \{0,1\}}{}{\ \ \ \forall j \in \{1,2,\ldots,n\}} \label{IntegerDomainConstraints}
\end{mini!}
Here, the decision variables are $x_1, x_2, \ldots, x_n$. The constraints (\ref{capacityConstraint1}), (\ref{IMSConstraints}) and (\ref{IntegerDomainConstraints}) impose the set of feasible solutions to be precisely the set of all IMSs. More specifically, constraint (\ref{capacityConstraint1}) represents the capacity constraint. Constraints (\ref{IMSConstraints}) impose the condition that no unselected item can be added to the knapsack without violating the capacity constraint. When $x_j=1$ the constraint is always satisfied and when $x_j=0$ it results in  $\sum_{i=1}^{n}(w_i x_i)+w_j \geq c+1$. Finally, constraints (\ref{IntegerDomainConstraints}) impose the decision variables to be binary.

We now derive a lower bound for the total weight of any IMS $\mathbf{x}=x_1, x_2, \ldots, x_n$ (i.e. any feasible solution for the previous optimization problem) . Since $\sum_{i=b}^n w_i > c$, there must be some $k \in \{b, b+1, \ldots, n\}$ for which $x_k=0$. This in turn implies that $\sum_{i=1}^{n} w_i x_i \geq c-w_b+1$, because $\mathbf{x}$ is an IMS and if $\sum_{i=1}^{n} w_i x_i \leq c-w_b$ would hold, we could set $x_k$ to 1 without exceeding the knapsack capacity.

The following linear optimization problem represents a relaxation of the previous one:
\begin{mini!}[2]
	{}{\sum_{i=1}^{n}p_i x_i}{}{} \label{objectiveFunction2}
	\addConstraint{\sum_{i=1}^{n}w_i x_i}{\leq c}{} \label{capacityConstraint2}
	\addConstraint{\sum_{i=1}^{n}w_i x_i}{\geq c-w_b+1}{} \label{lowerBoundConstraint}
	\addConstraint{0 \leq x_j \leq 1}{}{\ \ \ \forall j \in \{1,2,\ldots,n\}} \label{RealDomainConstraints}
\end{mini!}
In comparison with before, we dropped constraints (\ref{IMSConstraints}), we replaced the integrality constraints (\ref{IntegerDomainConstraints}) by a linear relaxation (constraints (\ref{RealDomainConstraints})) and we added the constraint (\ref{lowerBoundConstraint}) (which holds for any feasible solution of the original optimization problem) such that we indeed obtain a relaxation.

Since $p_i>0$ ($\forall i \in \{1,2,\ldots,n\}$) and we want to minimize $\sum_{i=1}^{n}p_i x_i=\sum_{i=1}^{n}\frac{p_i}{w_i}w_i x_i$, constraint (\ref{lowerBoundConstraint}) is tight for the optimal solution of the relaxation (i.e. we have $\sum_{i=1}^{n}w_i x_i=c-w_b+1$). This means that the optimal solution of the relaxation can be constructed by considering all items consecutively in non-decreasing order of $\frac{p_{\pi_i}}{w_{\pi_i}}$ and make $x_{\pi_i}$ as large as possible without violating constraints (\ref{RealDomainConstraints}) and stop when $\sum_{i=1}^{n}w_i x_i=c-w_b+1$. Hence, the optimal solution is given by setting $x_{\pi_j}=1$ ($\forall j \in \{1, 2, \ldots, f-1\}$), setting $x_{\pi_f}=\frac{c-w_b+1-\sum_{i=1}^{f-1}w_{\pi_i}}{w_{\pi_f}}$ and setting  $x_{\pi_j}=0$ ($\forall j \in \{f+1, f+2, ..., n\}$). The corresponding objective function value is given by $\sum_{i=1}^{f-1}(p_{\pi_i})+\frac{c-w_b+1-\sum_{i=1}^{f-1}w_{\pi_i}}{w_{\pi_f}} p_{\pi_f}$ such that for any IMS $\mathbf{x}$ we have:
$$\sum_{i=1}^n p_i x_i \geq \sum_{i=1}^{f-1}(p_{\pi_i})+\frac{c-w_b+1-\sum_{i=1}^{f-1}w_{\pi_i}}{w_{\pi_f}} p_{\pi_f}$$
\end{proof}

As a direct consequence of Theorem \ref{lowerBoundTheorem}, we obtain the first generalization of Theorem 1 from \cite{Jooken:2022} in Corollary \ref{ratioCorollary}, which relates the objective function value of an IMS to the optimal objective function value.
\begin{corollary}
\label{ratioCorollary}
Let $x_1, x_2, \ldots, x_n$ be an IMS, let $x_1^*, x_2^*, \ldots, x_n^*$ be an optimal solution for a given problem instance and let $\text{UB}$ be an upper bound for $\sum_{i=1}^n p_i x_i^*$. We have
$$ \frac{\sum_{i=1}^n p_i x_i}{\sum_{i=1}^n p_i x_{i}^*} \geq \frac{\sum_{i=1}^{f-1}(p_{\pi_i})+\frac{c-w_b+1-\sum_{i=1}^{f-1}w_{\pi_i}}{w_{\pi_f}} p_{\pi_f}}{\text{UB}}$$
\end{corollary}

Note that it is necessary to have an upper bound for the optimal objective function value to be able to calculate the right hand side of Corollary \ref{ratioCorollary}. Such an upper bound was indeed proposed for NMGE problem instances in \cite{Jooken:2022}, but it has also been a topic of active research for arbitrary 0-1 knapsack problem instances and several bounds are available in the literature \citep{Dantzig:1957, Martello:1977, Martello:1988, Martello:1999, Kellerer:2004a}.

Next, we consider a situation where the items can be divided into several groups depending on their weight. In Theorem \ref{cardinalityTheorem}, which proves a cardinality constraint for IMSs on the number of selected items in the last group, we generalize Theorem 2 from \cite{Jooken:2022}.

\begin{theorem}
\label{cardinalityTheorem}
Let $P=\{\{1,\ldots,m_1\},\{m_1+1, \ldots, m_2\},\ldots,\{m_{g-1}+1, \ldots, m_g\}\}$ (with $m_1 < m_2 < \ldots < m_g=n$ and $g \geq 2$) be a partition of the set of all items $\{1,2,\ldots,n\}$ into $g$ groups (i.e. the items within one group have consecutive indices). Define $m_0=0$ (for notational convenience). Let $\mathbf{x}=x_1, x_2, \ldots, x_n$ be an IMS, let $s_i=m_i-m_{i-1}$ be the number of items in group $i$ ($\forall i: 1 \leq i \leq g$) and let $l_i=\sum_{j=m_{i-1}+1}^{m_i} x_j$ be the number of selected items in $\mathbf{x}$ that belong to group $i$ ($\forall i: 1 \leq i \leq g$). We have:
$$l_g \geq z$$
where
$$z = \smin{} \setc{a \in \{0, 1, \ldots, s_g\}}{\sum_{i=m_{g-1}+1}^{m_{g-1}+a} w_i >  c - \smax{0 \leq n_i \leq s_i, \text{integer}\\ \sum_{i=1}^{g-1} w_{m_i} n_i \le c} \Big(\sum_{i=1}^{g-1} w_{m_{i-1}+1} n_i\Big) - w_{m_{g-1}+1}}$$
or $z=s_g$ if the above set over which the minimum is taken is empty.
\end{theorem}
\begin{proof}
The items are ordered in non-increasing order of weight and hence we have the following inequalities:
\begin{gather*}
\sum_{i=1}^{g-1} w_{m_{i}} l_i \leq \sum_{i=1}^{m_{g-1}} w_i x_i  \leq \sum_{i=1}^{g-1} w_{m_{i-1}+1} l_i
\end{gather*}
Since $\mathbf{x}$ is feasible, the $l_i$ must be such that they are integers for which $0 \leq l_i \leq s_i$ and $\sum_{i=1}^{g-1} w_{m_{i}} l_i \leq \sum_{i=1}^{m_{g-1}} w_i x_i \leq  \sum_{i=1}^{n} w_i x_i \leq c$. This means that $\sum_{i=1}^{g-1} w_{m_{i-1}+1} l_i$ is bounded from above by the maximum possible value it could attain while the restrictions on $l_i$ must apply:
$$\sum_{i=1}^{g-1} w_{m_{i-1}+1} l_i \leq \smax{0 \leq n_i \leq s_i, \text{integer}\\ \sum_{i=1}^{g-1} w_{m_i} n_i \le c} \sum_{i=1}^{g-1} w_{m_{i-1}+1} n_i$$
Here, the $n_i$ represent the integer decision variables of the maximization problem on the right hand side. Because of this, we also have:
\begin{align*}
c - \sum_{i=1}^{n} w_i x_i = c - \sum_{i=1}^{m_{g-1}} w_i x_i - \sum_{i=m_{g-1}+1}^{n} w_i x_i \geq c - \smax{0 \leq n_i \leq s_i, \text{integer}\\ \sum_{i=1}^{g-1} w_{m_i} n_i \le c} \Big(\sum_{i=1}^{g-1} w_{m_{i-1}+1} n_i\Big) - \sum_{i=m_{g-1}+1}^{n} w_i x_i 
\end{align*}
In case $l_g < s_g$, we must have $w_{m_{g-1}+1} > c - \sum_{i=1}^{n} w_i x_i$ because otherwise at least one item from $\{m_{g-1}+1, \ldots, n\}$ could be added without violating the capacity constraint. Hence we find that either $l_g=s_g$ or:
\begin{gather*}
w_{m_{g-1}+1} > c - \smax{0 \leq n_i \leq s_i, \text{integer}\\ \sum_{i=1}^{g-1} w_{m_i} n_i \le c} \Big(\sum_{i=1}^{g-1} w_{m_{i-1}+1} n_i\Big) - \sum_{i=m_{g-1}+1}^{n} w_i x_i 
\end{gather*}
Rewriting this gives:
\begin{gather*}
\sum_{i=m_{g-1}+1}^{n} w_i x_i >  c - \smax{0 \leq n_i \leq s_i, \text{integer}\\ \sum_{i=1}^{g-1} w_{m_i} n_i \le c} \Big(\sum_{i=1}^{g-1} w_{m_{i-1}+1} n_i\Big) - w_{m_{g-1}+1}
\end{gather*}
Since the items are ordered in non-increasing order of weight, we conclude that $l_g \geq z$ because of the above inequality and the theorem follows.
\end{proof}

Theorem \ref{cardinalityTheorem} assumes that the items are partitioned into $g$ groups based on their weights. There are several possibilities to make these groups and Theorem \ref{cardinalityTheorem} holds for arbitrary groups. In the current paper, we choose to make groups such that items within one group tend to have similar weights. More specifically, we model the problem of finding $g$ groups that fit the data well as a clustering problem \citep{Hartigan:1979} on the weights of the items. In this problem, we are interested in finding $g$ real values $\mu_1, \mu_2, \ldots, \mu_g$ (the centroids) such that the sum of the squared distances between a weight and its closest centroid is minimized. This means that we want to compute the following function $h(g)$:
$$h(g)=\smin{\mu_1, \mu_2, \ldots, \mu_g \in \mathbb{R}}\Big(\sum_{i=1}^{n} \smin{j \in \{1, 2, \ldots, g\}} (w_i-\mu_j)^2\Big)$$
together with the centroids $\mu_1, \mu_2, \ldots, \mu_g$ that minimize this function and the $g$ groups that these centroids induce. Any set of $g$ real numbers $\mu_1, \mu_2, \ldots, \mu_g$ partitions the set of all items $\{1, \ldots, n\}$ into a partition $P=\{\{1,\ldots,m_1\},\{m_1+1, \ldots, m_2\},\ldots,\{m_{g-1}+1, \ldots, m_g\}\}$ (i.e. the $g$ groups) with equivalence relation $\sim$, where items $i$ and $j$ belong to the same group (i.e. $i \sim j$) if and only if they have the same closest centroid (ties are broken in favour of the leftmost centroid):
$$i \sim j \iff  \arg\min_{k \in \{1, 2, \ldots, g\}} (w_i-\mu_k)^2=\arg\min_{k \in \{1, 2, \ldots, g\}} (w_j-\mu_k)^2$$
The function $h(g)$ is non-increasing, with $h(n)=0$. We are also interested in finding an appropriate number of groups $g$ such that the data fit well within $g$ groups (i.e. $h(g)$ is small), but the groups are fairly large (for $g=n$, every group is a singleton set). For typical problem instances, $h(g)$ first decreases very fast and then slowly goes towards 0. This motivates the introduction of $g^*$, which is the smallest value after which $h(g)$ does not decrease a lot anymore. More formally, we define $g^*$ as follows:
$$g^* = \min \setc{g \in \{2,\ldots,n-1\}}{\frac{h(g+1)}{h(g)} > \alpha}$$
Here, $\alpha$ is a parameter that denotes a real value between 0 and 1 and a higher value of $\alpha$ results in a higher $g^*$. This is visually illustrated in Fig. \ref{gStarFigure}, which depicts what $h(g)$ typically looks like and also depicts $g^*$ using a value of $\alpha=0.9$.
 	\begin{figure}[h]
	\centering
  \includegraphics[width=0.8\linewidth]{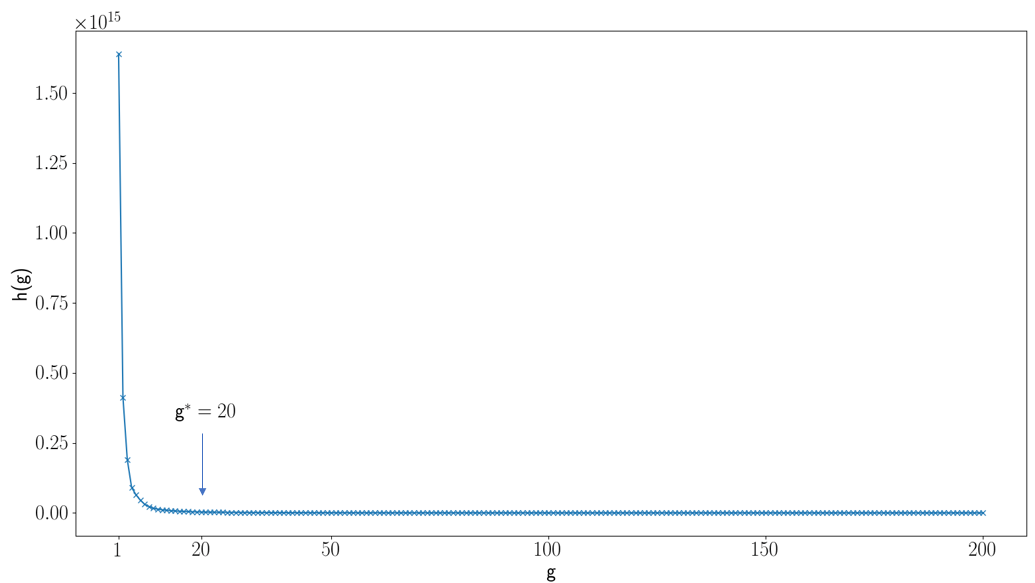}
\caption{A plot of $h(g)$ for a problem instance consisting of $n=200$ items. Using $\alpha=0.9$, the smallest value after which $h(g)$ does not decrease a lot anymore is $g^*=20$.}
\label{gStarFigure}       
\end{figure}
Using dynamic programming, it is possible to calculate $g^*$, $h(g^*)$, the centroids $\mu_1, \mu_2, \ldots, \mu_{g^*}$ that minimize $h(g^*)$ and the $g^*$ groups induced by these centroids with a time complexity of $O(n^2g^*)$ \citep{Wang:2011}.  This can be further optimized to $O(ng^*)$, because of the special structure of the dynamic programming formulation \citep{Gronlund:2017}. In the current paper, we implemented the $O(ng^*)$ algorithm.

The value of $g^*$ is typically small if the weights of the items occur in clusters. This also allows us to calculate $z$ from Theorem \ref{cardinalityTheorem} relatively efficiently. More specifically, the hardest problem to compute $z$ is to compute the following expression:
$$ \smax{0 \leq n_i \leq s_i, \text{integer}\\ \sum_{i=1}^{g-1} w_{m_i} n_i \le c} \Big(\sum_{i=1}^{g-1} w_{m_{i-1}+1} n_i\Big) $$
This expression can be computed by viewing it as a special case of the bounded knapsack problem. The bounded knapsack problem has been studied by several authors before and efficient algorithms are available in the literature \citep{Pisinger:2000,Andonov:2000,Becker:2019}. In the current paper, we used the \textit{\textbf{Bouknap}} algorithm from \cite{Pisinger:2000}. In a problem instance of the bounded knapsack problem, we are given a knapsack with a capacity of $c$ and $n$ different types of items. For item $i$ ($1 \leq i \leq n$), there are $s_i$ identical copies available and each copy has a profit of $p_i$ and a weight of $w_i$. The goal of the problem is to select a set of items such that the total profit is maximized and the total weight does not exceed the knapsack capacity. That is, we want to compute:
$$ \smax{0 \leq n_i \leq s_i, \text{integer}\\ \sum_{i=1}^{n} w_{m_i} n_i \le c} \Big(\sum_{i=1}^{n} p_i n_i\Big) $$
From this expression it now follows that computing
$$ \smax{0 \leq n_i \leq s_i, \text{integer}\\ \sum_{i=1}^{g-1} w_{m_i} n_i \le c} \Big(\sum_{i=1}^{g-1} w_{m_{i-1}+1} n_i\Big) $$
can be regarded as solving a bounded knapsack problem instance with $n=g-1$ and $p_i=w_{m_{i-1}+1}$ ($\forall i \in \{1, 2, \ldots, n\}$). Hence, we can reduce our problem to a standard problem for which efficient solvers exist \citep{Pisinger:2000,Andonov:2000,Becker:2019}.
\subsection{Discussion}
\label{discussion}
Since some of the features from the previous subsections can be extremely large for many practical problem instances (e.g. $\left|X\right|>10^{100}$ cannot be stored in a standard 64-bit integer), special care must be taken when implementing the algorithms. A straightforward way to deal with this, is to use data structures that support integer arithmetic on arbitrarily large integers. One of the biggest disadvantages of this approach, however, is that even simple computations can become slow. In our implementation of some of the algorithms, we used a different well-known approach when performing calculations with big numbers. We keep track of the binary logarithm of all numbers involved in the computations instead of the numbers themselves, which allows us to store all the necessary numbers in standard 64-bit floating point numbers. All involved operators can then be rewritten using standard logarithm rules (e.g. $\log(ab)=\log(a)+\log(b)$), which allows us to avoid working with big numbers directly. This implementation detail makes the algorithms fast enough to compute all features for reasonably large problem instances. A drawback of this approach is that we are unable to store very big numbers exactly (we only use 64 bits of precision), but this is not important for the purposes of this paper since we are mostly interested in the orders of magnitude and do not wish to distinguish cases where $\left|X\right|=10^{100}$ or $\left|X\right|=10^{100}+1$ for instance.

We also want to highlight that some features can sometimes still be very costly to compute, despite the fact that the algorithms that we propose are much faster than the straightforward exponential algorithms that would follow from directly using the definition of the features. For example, our algorithm to compute $\left|X\right|$ runs in $O(nc)$ time, which greatly improves the naive $O(n 2^n)$ approach (the naive approach explicitly builds $X$ by enumerating every solution and checking whether it is an IMS or not). As we mentioned before, computing the computationally challenging features from this paper can be more costly than computing the optimal solution for the given problem instance, since it is well-known that the 0-1 knapsack problem can be solved in $O(nc)$ time \citep{Bellman:1966} and several algorithms exist that behave much better in practice than their worst-case $O(nc)$ running time \citep{Pisinger:1995, Pisinger:1997, Martello:1999}. The cost of computing the features also affects how they can be used to answer different questions in different contexts and some features are more appropriate for one context, whereas other features are more appropriate for another context.

Some of the features that we propose are probably too costly to directly use in the context of algorithm selection without prior knowledge, where the features need to be computed for every problem instance before solving it and we do not have any prior knowledge about the problem instances. However, these features can be used in the context of algorithm selection if certain conditions are met. For instance, it is worthwhile investigating whether the proposed features can be approximated and computed more efficiently such that they become useful in the context of algorithm selection (regardless of prior knowledge). Furthermore, costly features can also be useful for algorithm selection scenarios where prior knowledge is available or applications where features are taken as a one-time input, but models are built repeatedly \citep{Hutter:2014}. One often deals with a scenario in which one needs to solve a large number of highly related problem instances that stem from the same application repeatedly (e.g. every hour) such that a fixed algorithm can be expected to behave similarly on all of these problem instances. Hence, in the case of algorithm selection with prior knowledge it is possible to learn an algorithm selector on a dataset consisting of more informative (and costly) features during the learning phase, because in the deployment phase the features only have to be computed for one or a handful of problem instances as opposed to the whole dataset.

Apart from algorithm selection scenarios, features also play a central role in problem instance space analysis. In this context, it is not always important how the cost of computing the features for a given problem instance relates to the cost of solving the problem instance to optimality, because the goal is to obtain deeper insight into the structure of the problem instance space. In this case more informative features (which are typically also more costly) are more useful. Instance space analysis has been used before to investigate a wide range of interesting questions such as  ``Where are the hard problem instances located?'', ``Which features correlate with problem instance hardness?'', ``How similar are two given problem instances?'' and ``Is a given dataset of problem instances varied enough to be a representative benchmark with regard to a given set of features?''. Hence, the features that we propose in the current paper are more suitable for the purpose of instance space analysis than for the purpose of algorithm selection (also recall from before that \textit{\textbf{Combo}} is almost always the fastest algorithm, which makes the benefit of using algorithm selection with the currently existing algorithms low anyhow).

The (features derived from) the bounds of the previous subsection are also of independent interest, because of their potential to be integrated into existing knapsack solvers. For instance, one of the key ideas for the most successful 0-1 knapsack algorithm from the literature (\textit{\textbf{Combo}}) was to add a valid cardinality constraint to the 0-1 knapsack problem without changing the optimum. This cardinality constraint can be surrogate relaxed with the capacity constraint and its linear relaxation can be efficiently solved and yields strong dual bounds. Since the optimal solution to the 0-1 knapsack problem is always an IMS, the bounds derived in Subsection \ref{boundsFeatures} can also be added to the problem without changing the optimum and this could be interesting to obtain strong relaxations. We encourage further research to investigate which relaxations would work best (e.g. linear relaxations, surrogate relaxations or Lagrangian-based relaxations) and how these relaxations could be efficiently solved.
	\section{Experiments}
\label{experimentsSection}
\subsection{Control experiment: correctness verification and sanity checks}
The algorithms that we propose in Section \ref{featureSection} to calculate the different features have several tricky implementation details (see the discussion in Subsection \ref{discussion}). We already formally proved that these algorithms are correct in Section \ref{featureSection}, but this does not rule out potential implementation bugs. Apart from carefully implementing these algorithms, we also systematically verified that we did not introduce any bugs by comparing the output of our algorithms with the output of several brute force algorithms that we also implemented. These brute force algorithms are very slow, but very easy to implement and are useful tools to experimentally verify the correctness of faster, but harder algorithms. More specifically, for the features described in Subsection \ref{countingBasedFeatures} we also implemented a brute force algorithm that explicitly generates the set of all IMSs $X$ by considering all $2^n$ solutions and filtering out the IMSs. The features were then calculated by directly using the definition and compared with the features that were calculated by the dynamic programming algorithm from Subsection \ref{dynamicProgramming}. The brute force algorithm also allowed us to do a sanity check to verify that the bound that we obtained in Theorem \ref{lowerBoundTheorem} is correct, by checking the inequality for every IMS in $X$. To do a sanity check for Theorem \ref{cardinalityTheorem}, we implemented another brute force algorithm that generates every partition of the set of items $\{1, 2, \ldots, n\}$ and checks the inequality for every IMS. A slight modification of this algorithm was also used to calculate $h(g)$ and $g^*$ from Subsection \ref{boundsFeatures} (the optimal centroids that minimize $h(g)$ are given by the average of the weights of each group induced by the partition \citep{Gronlund:2017}). The output of all brute force algorithms was compared with the output of the algorithms proposed in Section \ref{featureSection} for $10^4$ small instances that were generated by choosing the number of items $n$ uniformly at random between 5 and 12, choosing the knapsack capacity $c$ uniformly at random between 2 and $10^8$ and choosing the profits and weights of each item uniformly at random between 1 and $c$. Through this control experiment we were able to experimentally verify that the results obtained by our algorithms matched the expected results, as desired.

\subsection{Influence of the features on the running time and a comparison with other features}
In this subsection, we show that the features that we propose in the current paper strongly affect the running time of the state-of-the-art knapsack algorithm \textit{\textbf{Combo}} \citep{Martello:1999}. We also show that these features are more informative for the running time (and more costly to compute) than other features from the literature. For this experiment, we used two datasets of problem instances from the literature. The first dataset (dataset A) was obtained by taking the 3000 most difficult problem instances from \cite{Pisinger:2005} and the second dataset (dataset B) consists of problem instances from a recently proposed class of hard problem instances \citep{Jooken:2022}. For both datasets we filtered out the problem instances with a knapsack capacity that exceeds $10^8$ (to ensure that the computational cost of calculating the features is more manageable), resulting in 1422 problem instances out of 3000 for dataset A and 2160 problem instances out of 3240 for dataset B. For these problem instances, the number of items $n$ varies between 50 and $10^4$ and the knapsack capacity $c$ varies between 40264 and $10^8$. For both datasets, we calculated all features from the current paper (see Table \ref{featuresOverviewTable}), all features from \cite{Smith-Miles:2021} and we solved all problem instances to optimality with \textit{\textbf{Combo}}. This experiment was conducted on the ThinKing cluster of the Flemish Supercomputer Center (VSC), and took around 540 CPU-hours using powerful CPUs with a clock rate of 2.5 GHz and 10 GB RAM memory. In this context, we also briefly mention the recent dataset of problem instances from \cite{Smith-Miles:2021}. We considered also using this dataset, but finally decided not to use it because the problem instances from this dataset could be solved in a couple of milliseconds (i.e. they are not hard) and we were unable to reliably measure such small running times, leading to a very low signal-to-noise ratio.

After calculating the raw features, we normalized them to make the features comparable in magnitude and reasonably small. The features from \cite{Smith-Miles:2021} were already normalized, whereas the features from the current paper (see Table \ref{featuresOverviewTable}) were normalized as follows. For the features $|X|$, $t_1$, $t_2$ and $t_3$ (as well as the dependent variable: the running time of \textit{\textbf{Combo}}) we took the (binary) logarithm (because their domains were either unbounded or very large) and for all other features we applied min-max normalization. Using these normalized features, we trained XGBoost \citep{Chen:2016} (a state-of-the-art regression model) to predict (the binary logarithm of) the running time of \textit{\textbf{Combo}} based on various sets of features leading to various models. Note that the goal of this experiment is to get new insights into what easy and hard problem instances look like and which features are related to the hardness of a problem instance (if we would only care about the running time of \textit{\textbf{Combo}}, it would be more efficient to just run \textit{\textbf{Combo}}). Both datasets were randomly split into a training set to train the model (consisting of 80\% of the data) and a test set (consisting of the remaining 20\% of the data) to test the results of the model on unseen data. The hyperparameters of each XGBoost model were tuned by using grid search in combination with 10-fold cross-validation.

In the first experiment, we trained three different XGBoost models for each dataset: (i) a model that uses the features from \cite{Smith-Miles:2021}, (ii) a model that uses the features from the current paper, and (iii) a model that uses the features from both \cite{Smith-Miles:2021} and the current paper. To put the obtained results into perspective, we also compared this model against a baseline model, which always predicts a constant value that minimizes the sum of the squared errors on the training set (the minimizer of this function is in fact the average of the dependent variable over the training set). More formally, the baseline model predicts the following value:
$$\argmin_{x \in \mathbb{R}} \Bigg( \sum_{i=1}^{\left| \text{train} \right|} (x-y_i)^2 \Bigg) = \frac{\sum_{i=1}^{\left| \text{train} \right|} y_i }{\left| \text{train} \right|}$$
Here, $\left| \text{train} \right|$ represents the size of the training set and the $y_i$'s represent the dependent variable (the binary logarithm of the running time of \textit{\textbf{Combo}}). If we denote by $\left| \text{test} \right|$ the size of the test set, by $y_i$ observation $i$ of the dependent variable, by $\hat{y_i}$ the prediction for observation $i$ and by $\bar{y}=\argmin_{x \in \mathbb{R}} \Bigg( \sum_{i=1}^{\left| \text{test} \right|} (x-y_i)^2 \Bigg)=\frac{\sum_{i=1}^{\left| \text{test} \right|}y_i}{\left| \text{test} \right|}$ the average of the dependent variable over the test set, then the mean squared error is defined as $\text{MSE}=\frac{\sum_{i=1}^{\left| \text{test} \right|} (\hat{y_i}-y_i)^2}{\left| \text{test} \right|}$ and the $R^2$ value is defined as $R^2 = 1-\frac{\sum_{i=1}^{\left| \text{test} \right|} (\hat{y_i}-y_i)^2}{\sum_{i=1}^{\left| \text{test} \right|} (\bar{y_i}-y_i)^2}$ (this measures the proportion of the variation of the dependent variable that the model can explain based on the independent variables). Hence, lower MSE values and higher $R^2$ values indicate a better performing model.

In Table \ref{datasetPisinger} and Table \ref{datasetJooken} we show the mean squared error (MSE) and the $R^2$ value on the test set of dataset A and dataset B respectively for both the baseline model and each of the three XGBoost models. For both datasets, we see that both the features proposed in \cite{Smith-Miles:2021} and the features from the current paper are important indicators of the running time of \textit{\textbf{Combo}}. The MSEs of all three XGBoost models are considerably smaller than the MSE of the baseline model and analogously the $R^2$ values are considerably higher. If we compare the three XGBoost models for dataset A, we see that the model which uses the features from \cite{Smith-Miles:2021} has less predictive power than the model which uses the features from the current paper. Furthermore, for dataset A both models have less predictive power than the model which uses both sets of features and this indicates that both sets of features contain useful information that is not present in the other set of features. The MSEs and $R^2$ values of these three models are also quite different. The MSEs are 1.582, 1.034 and 0.843 and the $R^2$ values are 0.422, 0.623 and 0.692 for model (i), (ii) and (iii) respectively. For dataset B, the models are ranked from having least predictive power to most predictive power as follows: first model (i), then model (iii) and finally model (ii). For this dataset, the MSEs and $R^2$ values of the three XGBoost models are more similar than for dataset A and the predictive power of the three models is quite comparable. The MSEs for dataset B are 1.868, 1.727 and 1.821 and the $R^2$ values are 0.804, 0.819 and 0.809 for model (i), (ii) and (iii) respectively. The proportion of the variation of the dependent variable that the model can explain by using the independent variables (i.e. the $R^2$ value) is considerably closer to 1 for dataset B than for dataset A (a model that would always perfectly predict the dependent variable would have an $R^2$ value of 1). For both datasets we conclude that the features that we propose in the current paper strongly affect the running time of \textit{\textbf{Combo}} and that the XGBoost model with the most predictive power is the one that uses the features from the current paper (together with the features of \cite{Smith-Miles:2021} for dataset A and in isolation for dataset B). As we mentioned before, the informativeness of the features from the current paper also comes at an increased computing cost (see Table \ref{featureTimesTable}). However, recall from before that we are interested in assessing whether the features are important hardness indicators or not and the running times of calculating the features are less important for assessing the informativeness of the features.

\begin{table}[h!] \scriptsize\centering 
	\begin{threeparttable}
		\caption{Mean squared error and $R^2$ value on the test set of dataset A \citep{Pisinger:2005} for various models.}
		\label{datasetPisinger} 
		\begin{tabular}{cccccccccccccccccccccc} \\
			\hline\noalign{\smallskip} 
			\multicolumn{1}{c}{} & \multicolumn{2}{c}{Features}\\
			\cmidrule(lr){2-3}
			\noalign{\smallskip} 
			Model & \cite{Smith-Miles:2021} & Current paper & MSE & $R^2$\\ 
			\noalign{\smallskip}
			\hline
			\noalign{\smallskip} 
			\multicolumn{1}{c}{Baseline} & \xmark & \xmark & 2.752 & -0.005 \\ 	
			\noalign{\smallskip}
			\hline
			\noalign{\smallskip}
			\multicolumn{1}{c}{XGBoost} & \cmark & \xmark & 1.582 & 0.422 \\ 
			\multicolumn{1}{c}{XGBoost} & \xmark & \cmark & 1.034 & 0.623 \\ 
			\multicolumn{1}{c}{XGBoost} & \cmark & \cmark & 0.843 & 0.692 \\ 
			\noalign{\smallskip}
			\hline
		\end{tabular}
	\end{threeparttable}
\end{table}

\begin{table}[h!] \scriptsize\centering 
	\begin{threeparttable}
		\caption{Mean squared error and $R^2$ value on the test set of dataset B \citep{Jooken:2022} for various models.}
		\label{datasetJooken} 
		\begin{tabular}{cccccccccccccccccccccc} \\
			\hline\noalign{\smallskip} 
			\multicolumn{1}{c}{} & \multicolumn{2}{c}{Features}\\
			\cmidrule(lr){2-3}
			\noalign{\smallskip} 
			Model & \cite{Smith-Miles:2021} & Current paper & MSE & $R^2$\\ 
			\noalign{\smallskip}
			\hline
			\noalign{\smallskip} 
			\multicolumn{1}{c}{Baseline} & \xmark & \xmark & 9.549 & -0.003 \\ 	
			\noalign{\smallskip}
			\hline
			\noalign{\smallskip}
			\multicolumn{1}{c}{XGBoost} & \cmark & \xmark & 1.868 & 0.804 \\ 
			\multicolumn{1}{c}{XGBoost} & \xmark & \cmark & 1.727 & 0.819 \\ 
			\multicolumn{1}{c}{XGBoost} & \cmark & \cmark & 1.821 & 0.809 \\ 
			\noalign{\smallskip}
			\hline
		\end{tabular}
	\end{threeparttable}
\end{table}

\begin{table}[h!] \scriptsize\centering 
	\begin{threeparttable}
		\caption{An overview of the total running times for different algorithms that were used to compute the features and the dependent variable.}
		\label{featureTimesTable} 
		\begin{tabular}{cccccccccccccccccccccc} \\
			\hline\noalign{\smallskip} 
			dataset & Algorithm & Total running time \\ 
			\noalign{\smallskip}
			\hline
			\noalign{\smallskip} 
			\multicolumn{1}{c}{A: \cite{Pisinger:2005}} & \cite{Smith-Miles:2021} & 0.09 hours\\ 
			\multicolumn{1}{c}{A: \cite{Pisinger:2005}} & Dynamic programming (Subsection \ref{dynamicProgramming}) & 230.63 hours\\ 
			\multicolumn{1}{c}{A: \cite{Pisinger:2005}} & \cite{Gronlund:2017} & 18.56 hours \\
			\multicolumn{1}{c}{A: \cite{Pisinger:2005}} & \cite{Pisinger:2000} & 0.25 hours\\
			\multicolumn{1}{c}{A: \cite{Pisinger:2005}} & \cite{Martello:1999} & 0.19 hours\\
			\multicolumn{1}{c}{B: \cite{Jooken:2022}} & \cite{Smith-Miles:2021} & 0.04 hours \\ 
			\multicolumn{1}{c}{B: \cite{Jooken:2022}} & Dynamic programming (Subsection \ref{dynamicProgramming})  &  283.68 hours\\ 
			\multicolumn{1}{c}{B: \cite{Jooken:2022}} & \cite{Gronlund:2017} & 0.04 hours \\
			\multicolumn{1}{c}{B: \cite{Jooken:2022}} & \cite{Pisinger:2000} & 0.22 hours \\
			\multicolumn{1}{c}{B: \cite{Jooken:2022}} & \cite{Martello:1999} & 7.84 hours\\    
			\noalign{\smallskip}
			\hline
		\end{tabular}
	\end{threeparttable}
\end{table}

\subsection{Identifying the most informative features}
\label{mostInformativeFeaturesSubsection}
In the second experiment, we try to identify the most informative features for the XGBoost model to predict the running time of \textit{\textbf{Combo}}. In this experiment we consider 58 features: the 44 features proposed in \cite{Smith-Miles:2021} and the 14 features proposed in the current paper. We trained 58 different XGBoost models; each model was trained using all 58 features except for one feature. If omitting a specific feature results in an increased MSE, then we can conclude that this feature is likely to be important since the XGBoost model performs worse when this feature is not available. It is also possible that omitting a certain feature results in a decreased MSE, but then it is more difficult to draw a sound conclusion. In the latter case, it is possible that the feature is not important to predict the dependent variable, but it could also be the case that this feature is highly correlated with another independent variable such that the marginal contribution of this feature is low (e.g. a set of two or more very similar features can be very important to predict the dependent variable, but omitting one of them does not have a big effect). The results of this experiment can be found in Table \ref{omitExperimentPisinger} and Table \ref{omitExperimentJooken} for dataset A and B respectively. For each dataset, we ranked the 58 XGBoost models from high MSE to low MSE and we show the MSEs of the models corresponding to the six highest ranked and two lowest ranked features (as well as the reference model which uses all features).

For both datasets, we can see that the MSE increases considerably (and the $R^2$ value decreases considerably) when a highly ranked feature is omitted. Interestingly, we also see that several features obtain a high rank for both datasets (e.g. $t_1$, $t_3$ and $\smin{\mathbf{x} \in X} \text{totalWeight}(\mathbf{x})$ from the current paper). This indicates that these features can be expected to be informative for predicting the running time of \textit{\textbf{Combo}}. The running times $t_1$ and $t_3$ are highly ranked for both datasets and $t_2$ is highly ranked for dataset A (rank 2), but not so highly ranked for dataset B (rank 46). They measure the running times for computing various features related to IMSs (see Table \ref{featuresOverviewTable}) and this confirms that the structure of IMSs can be an important indicator of the hardness of a problem instance, despite the fact that these three algorithms are all based on very different ideas than \textit{\textbf{Combo}}.

\begin{table}[h!] \scriptsize\centering 
	\begin{threeparttable}
		\caption{Models trained on all features except for one feature, ranked according to MSE for the test set of dataset A \citep{Pisinger:2005}}
		\label{omitExperimentPisinger} 
		\begin{tabular}{cccccccccccccccccccccc} \\
			\hline
			\noalign{\smallskip} 
			Features used by XGBoost & Rank & MSE & $R^2$\\ 
			\noalign{\smallskip}
			\hline
			\noalign{\smallskip}
			\multicolumn{1}{c}{All features} & NA & 0.843 & 0.692\\ 
			\noalign{\smallskip}
			\hline
			\noalign{\smallskip}
			\multicolumn{1}{c}{All features except}\\
			\multicolumn{1}{l}{Current paper: $t_1$} & 1 & 1.084 & 0.604\\ 
			\multicolumn{1}{l}{Current paper: $t_2$} & 2 & 0.967 & 0.647\\
			\multicolumn{1}{l}{Current paper: $\smin{\mathbf{x} \in X} \text{totalWeight}(\mathbf{x})$} & 3 & 0.925 & 0.662\\ 
			\multicolumn{1}{l}{\cite{Smith-Miles:2021}: First Weight} & 4 & 0.889 & 0.675\\
			\multicolumn{1}{l}{\cite{Smith-Miles:2021}: Reduced Polyfit Linear} & 5 & 0.889 & 0.675\\ 
			\multicolumn{1}{l}{Current paper: $t_3$} & 6 & 0.889 & 0.675\\
			\multicolumn{1}{l}{...} & ... & ... & ...\\
			\multicolumn{1}{l}{\cite{Smith-Miles:2021}: Coefficient of Variation of Efficiencies} & 57 & 0.828 & 0.698\\ 
			\multicolumn{1}{l}{\cite{Smith-Miles:2021}: Coefficient of Variation of Weights} & 58 & 0.782 & 0.714\\ 
			\noalign{\smallskip}
			\hline
		\end{tabular}
	\end{threeparttable}
\end{table}

\begin{table}[h!] \scriptsize\centering 
	\begin{threeparttable}
		\caption{Models trained on all features except for one feature, ranked according to MSE for the test set of dataset B \citep{Jooken:2022}}
		\label{omitExperimentJooken} 
		\begin{tabular}{cccccccccccccccccccccc} \\
			\hline
			\noalign{\smallskip} 
			Features used by XGBoost & Rank & MSE & $R^2$\\ 
			\noalign{\smallskip}
			\hline
			\noalign{\smallskip}
			\multicolumn{1}{c}{All features} & NA & 1.821 & 0.809\\ 
			\noalign{\smallskip}
			\hline
			\noalign{\smallskip}
			\multicolumn{1}{c}{All features except}\\
			\multicolumn{1}{l}{Current paper: $t_1$} & 1 & 2.031 & 0.787\\ 
			\multicolumn{1}{l}{Current paper: $t_3$} & 2 & 1.965 & 0.794\\
			\multicolumn{1}{l}{Current paper: $f$} & 3 & 1.878 & 0.803\\ 
			\multicolumn{1}{l}{\cite{Smith-Miles:2021}: Reduced Maximum Cardinality} & 4 & 1.856 & 0.805\\
			\multicolumn{1}{l}{Current paper: $\smin{\mathbf{x} \in X} \text{totalWeight}(\mathbf{x})$} & 5 & 1.848 & 0.806\\ 
			\multicolumn{1}{l}{\cite{Smith-Miles:2021}: Smaller Better Pairs} & 6 & 1.842 & 0.807\\
			\multicolumn{1}{l}{...} & ... & ... & ...\\
			\multicolumn{1}{l}{\cite{Smith-Miles:2021}: Reduced Coefficient of Variation of Efficiencies} & 57 & 1.751 & 0.816\\ 
			\multicolumn{1}{l}{\cite{Smith-Miles:2021}: Dominant Pairs} & 58 & 1.744 & 0.817\\ 
			\noalign{\smallskip}
			\hline
		\end{tabular}
	\end{threeparttable}
\end{table}

\subsection{Instance space analysis}

In the last experiment, we use the Instance Space Analysis methodology developed by \cite{Smith-Miles:2012,Smith-Miles:2014,Smith-Miles:2015} to visualize the problem instances in a two-dimensional space. This allows us to investigate where the hard problem instances are and how the features are distributed in the easy and hard regions of the problem instance space. In this methodology, the problem instances are represented by points in a high-dimensional feature space (where the dimension is equal to the number of features). These points are then normalized and projected to a two-dimensional instance space by finding two appropriate linear combinations of these features. We use the tool MATILDA \citep{Smith-Miles:2019} which implements the Instance Space Analysis methodology and uses the PILOT method \citep{Munoz:2018} to find this projection. The PILOT method attempts to find a projection such that both the features and the algorithm's performance (in our case the running time of \textit{\textbf{Combo}}) are distributed in a nearly linear fashion across the instance space. This makes the visualization highly interpretable and useful to visually see patterns connecting the features with problem instance hardness. For this experiment, we used the problem instances from dataset A and B together with the 6 highest ranked features from Table \ref{omitExperimentPisinger} and the 6 highest ranked features from Table \ref{omitExperimentJooken}. There are 3 features in common ($t_1$, $t_3$ and $\smin{\mathbf{x} \in X} \text{totalWeight}(\mathbf{x})$), leading to a total of $6+6-3=9$ features (less than 10 features are recommended for MATILDA). The projection matrix that we obtain for projecting the problem instances as points in a 9D feature space to a 2D instance space with axes $z_1$ and $z_2$ is as follows:

\begin{equation*}
\begin{bmatrix}
z_1\\
z_2\\
\end{bmatrix}
 = 
\begin{bmatrix}
0.2899 & -0.2316 \\
0.2924 & -0.2034 \\
-0.3407 & -0.2515 \\
0.1679 & -0.5357 \\
0.2802 & -0.3762 \\
-0.0796 & -0.5672 \\
0.4686 & 0.6227 \\
0.3397 & -0.3426 \\
-0.5083 & 0.0148 \\
\end{bmatrix}^T
\begin{bmatrix}
t_1 \\
t_2  \\
t_3  \\
f  \\
\smin{\mathbf{x} \in X} \text{totalWeight}(\mathbf{x}) \\
\text{First Weight} \\
\text{Smaller Better Pairs}  \\
\text{Reduced Maximum Cardinality} \\
\text{Reduced Polyfit Linear} \\
\end{bmatrix}
\end{equation*}

The projected problem instances can be found in Fig. \ref{sourcesFigure}, where each problem instance is labelled according to the dataset it is in. This figure reveals that the problem instances from both datasets are quite different with respect to the features. It is remarkable that almost all problem instances from both datasets can be found in two disjoint parts of the instance space, although the PILOT method did not explicitly have this as a constraint for the projection matrix. The problem instances from dataset B occur in a banana-shaped region with extremities around $(z_1,z_2)=(-2,4)$ and $(z_1,z_2)=(1,-1)$, whereas almost all problem instances from dataset A are spread throughout the instance space, but only a few of them are located in this banana-shaped region.
 	\begin{figure}[H]
	\centering
  \includegraphics[width=0.6\linewidth]{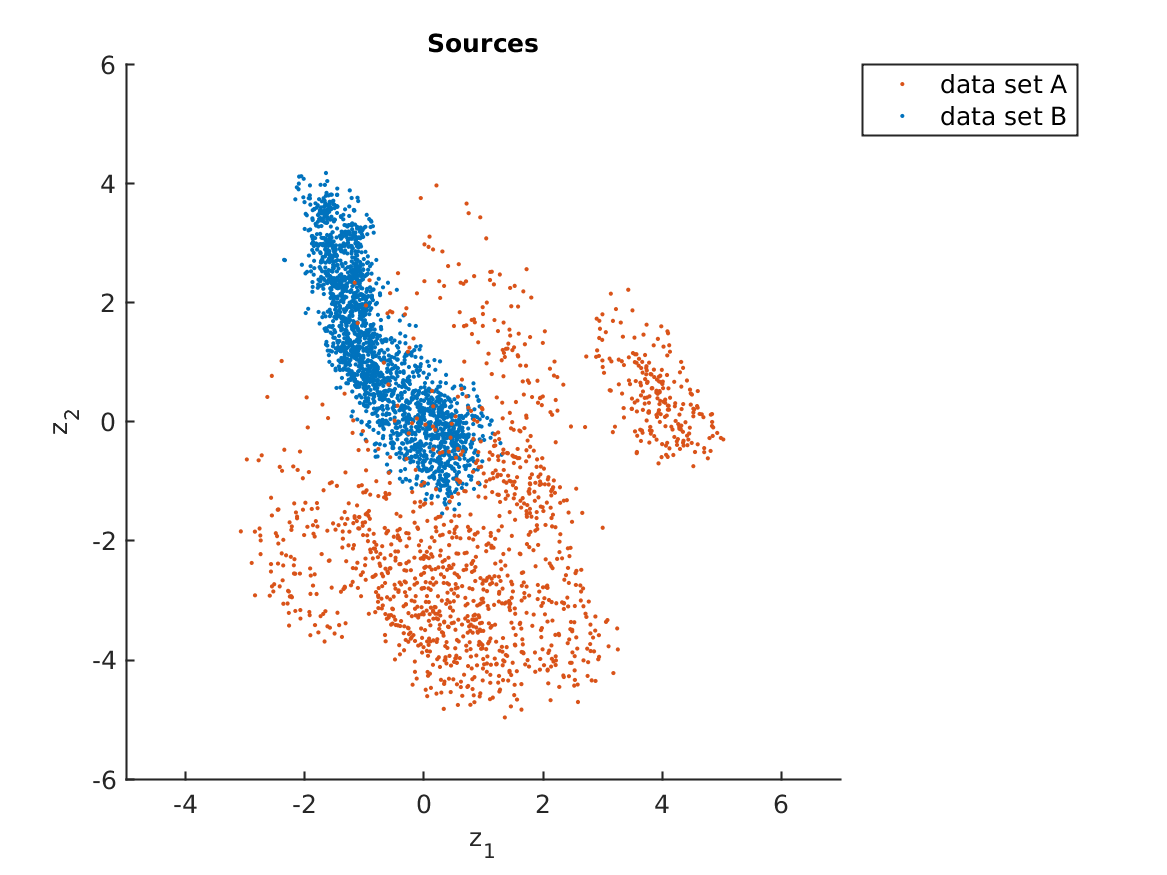}
\caption{The problem instances from dataset A and B visualized in a 2D problem instance space.}
\label{sourcesFigure}       
\end{figure}

 	\begin{figure}[H]
	\centering
  \includegraphics[width=0.6\linewidth]{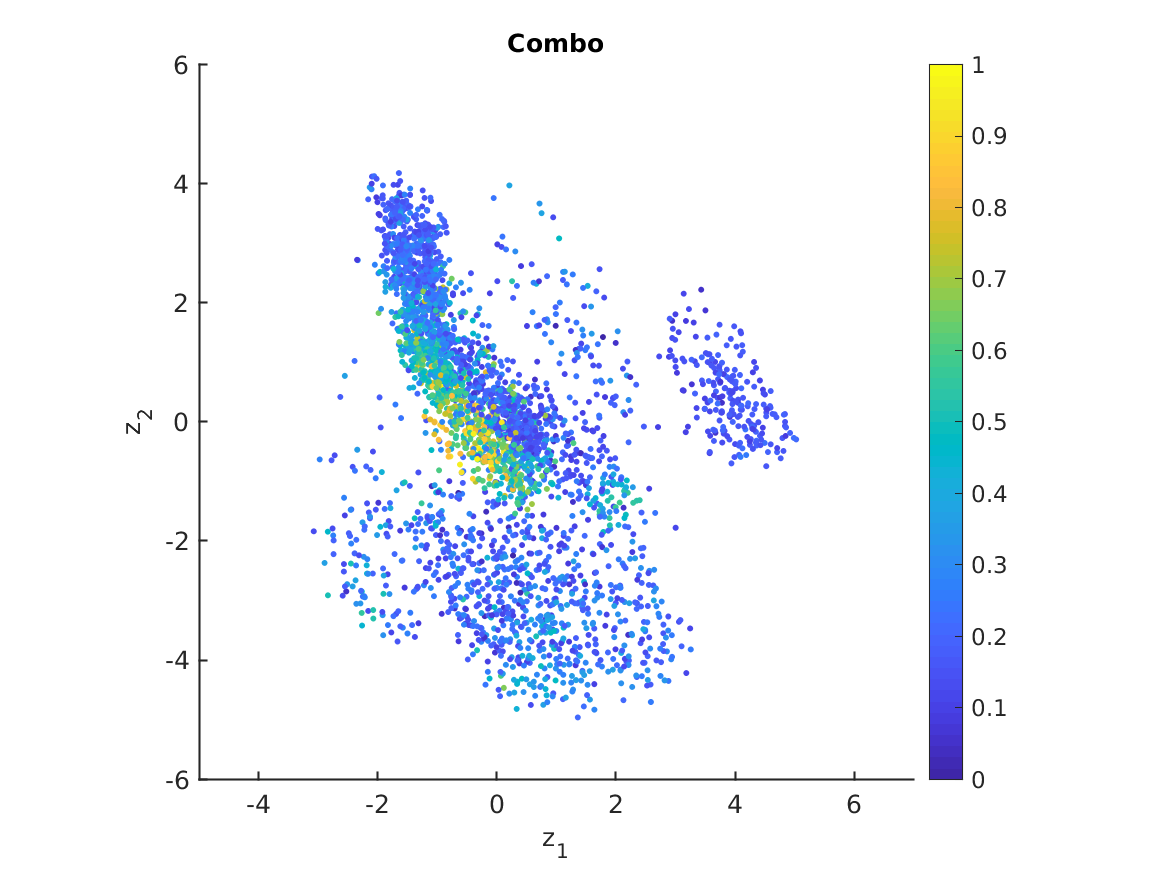}
\caption{The normalized running times of \textit{\textbf{Combo}} visualized in a 2D problem instance space.}
\label{ComboPerformanceFigure}       
\end{figure}

In Fig. \ref{ComboPerformanceFigure}, we can find the binary logarithm of the running time of \textit{\textbf{Combo}} (normalized between 0 and 1) visualized in the same instance space. As we can see, the most difficult problem instances occur in the same region of the instance space near $(z_1,z_2)=(-0.3,-0.3)$. Almost all of these hard problem instances come from dataset B, which is consistent with Table \ref{featureTimesTable}. The hardest problem instances from dataset A on the contrary occur in different regions of the instance space (e.g. one region is near $(z_1,z_2)=(2,-1)$ and another region is near $(z_1,z_2)=(1,-4)$). We also visualized the distribution of the 9 normalized features in Fig. \ref{9FeaturesFigure} (larger versions have been added to the Appendix for the readers' convenience in Figs. \ref{t1Figure}-\ref{reducedPolyfitLinearFigure}). This allows us to see that several features have a different distribution in the hard and easy regions of the instance space, which makes these features important indicators of the hardness of a problem instance. For instance, the features $t_1$ and $t_3$ (Figs. \ref{t1Subfigure} and \ref{t3Subfigure}) tend to be very high near $(z_1,z_2)=(-0.3,-0.3)$ where most hard problem instances are located. This means that, if a problem instance is hard to solve for \textit{\textbf{Combo}}, then typically it is also hard to calculate several features related to the IMSs of that problem instance. Note that the reverse is not always true: there are several places in the instance space for which $t_1$ and $t_3$ are high, but the running time of \textit{\textbf{Combo}} is low. The features $t_2$, First Weight and Reduced Polyfit Linear (Figs. \ref{t2Subfigure}, \ref{firstWeightSubfigure} and \ref{reducedPolyfitLinearSubfigure}) are almost always the same (around 0.6, 0.0 and 1.0 respectively) for the problem instances of dataset B, but they vary more for the problem instances of dataset A. This also explains the outcome of the experiment from Subsection \ref{mostInformativeFeaturesSubsection} in which we saw that these features were very informative to predict the running time of \textit{\textbf{Combo}} for dataset A, but less informative for dataset B. The features $f$, $\smin{\mathbf{x} \in X} \text{totalWeight}(\mathbf{x})$ and Reduced Maximum Cardinality (Fig. \ref{fSubfigure}, \ref{minWeightSubfigure} and \ref{reducedMaximumCardinalitySubfigure}) increase when going from the top left corner of the instance space to the bottom right corner of the instance space. This increase occurs in a nearly linear fashion for the features $f$ and Reduced Maximum Cardinality, whereas the transition is more abrupt for the feature $\smin{\mathbf{x} \in X} \text{totalWeight}(\mathbf{x})$. The feature Smaller Better Pairs (Fig. \ref{smallerBetterPairsSubfigure}) on the other hand increases in a nearly linear fashion when going from the bottom left corner of the instance space to the top right corner of the instance space. Since we can see that these features behave differently in different parts of the instance space, the combined behaviour of all the features together is useful to distinguish the different parts of the instance space and this visually explains why the features are useful to assess the hardness of a problem instance.

\begin{figure}[H]
\centering
\begin{subfigure}{0.25\linewidth}
\centering\includegraphics[width=1.1\linewidth]{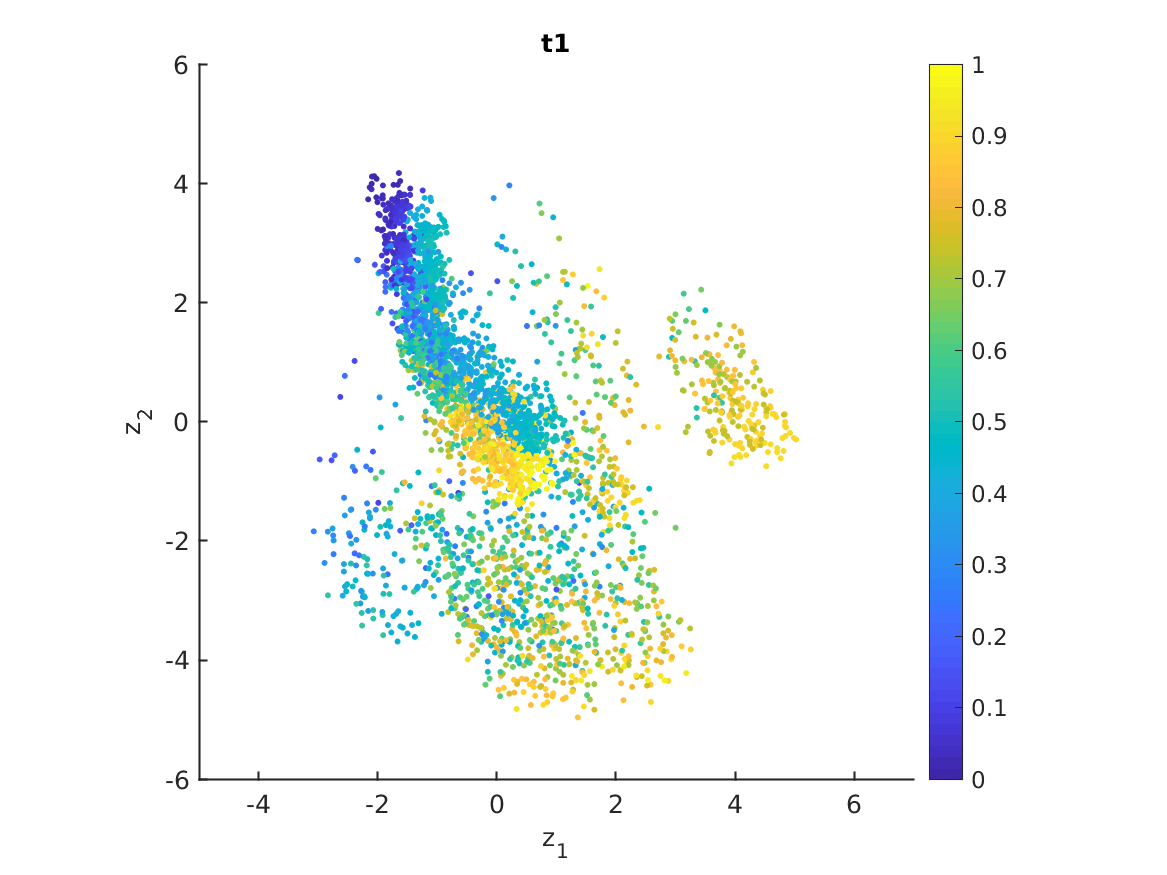}
\caption{}
\label{t1Subfigure}
\end{subfigure}%
\begin{subfigure}{0.25\linewidth}
\centering\includegraphics[width=1.1\linewidth]{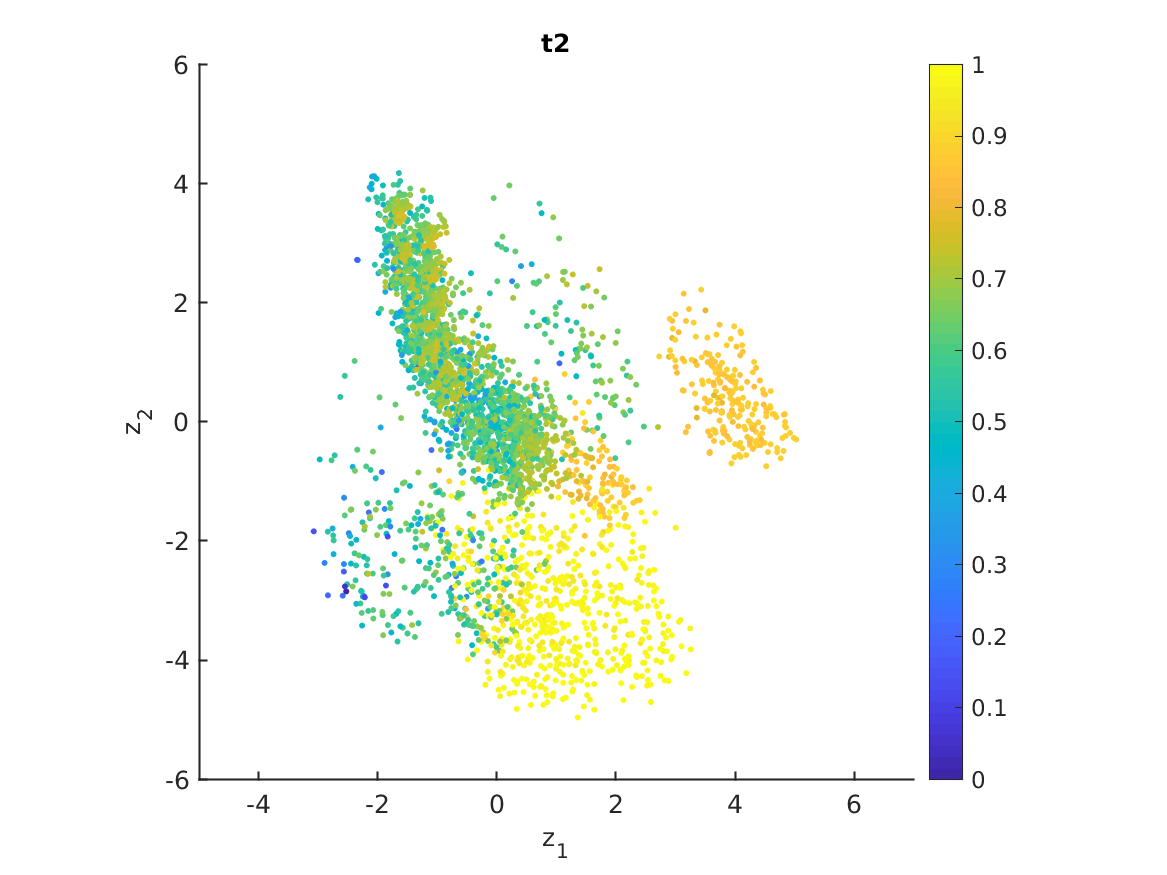}
\caption{}
\label{t2Subfigure}
\end{subfigure}
 \begin{subfigure}{0.25\linewidth}
\centering\includegraphics[width=1.1\linewidth]{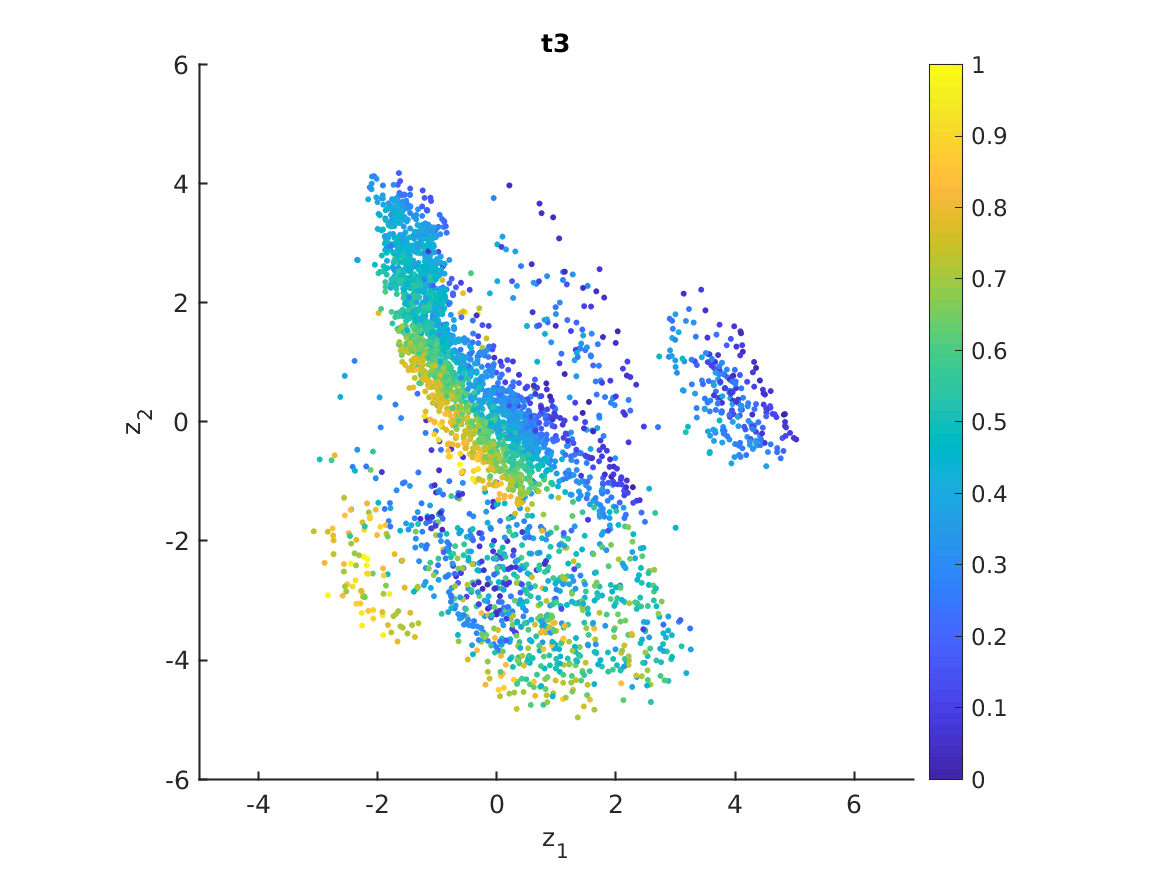}
\caption{}
\label{t3Subfigure}
\end{subfigure}\vspace{10pt}

\begin{subfigure}{0.25\linewidth}
\centering\includegraphics[width=1.1\linewidth]{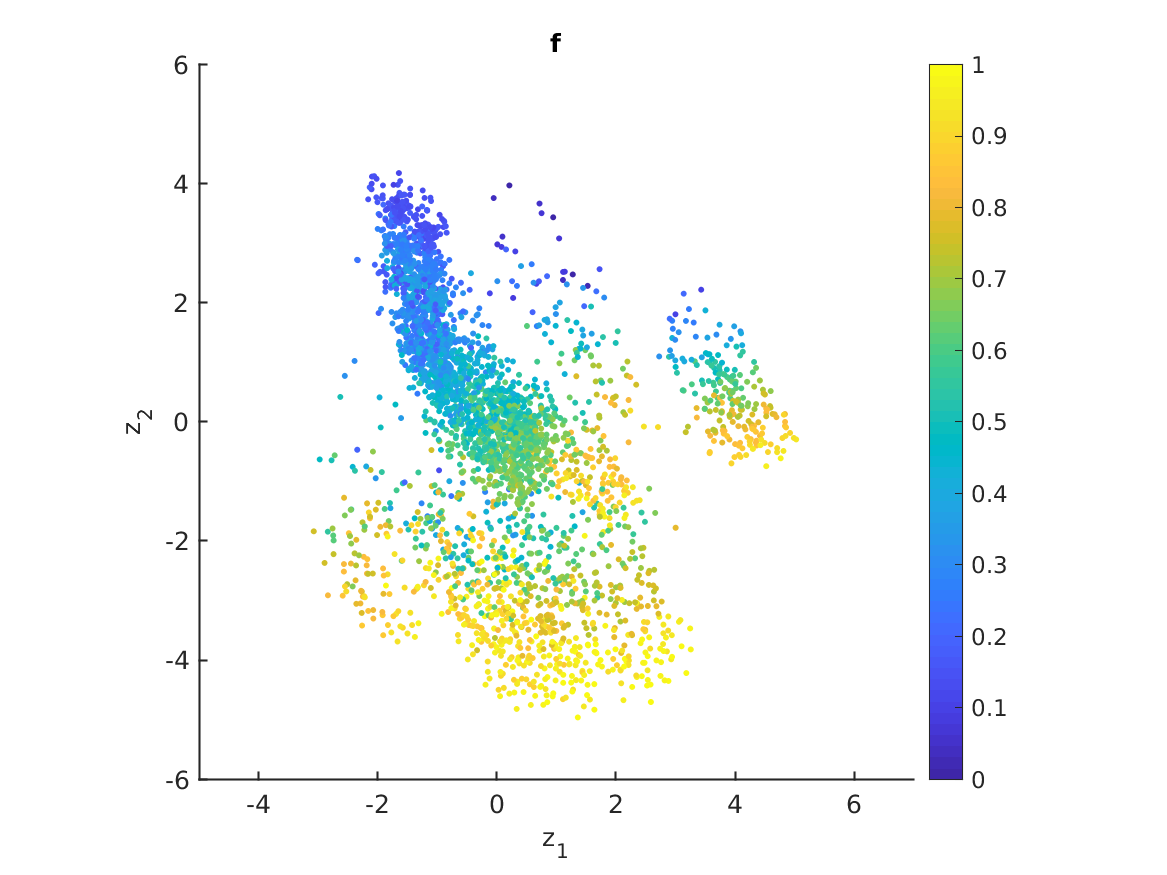}
\caption{}
\label{fSubfigure}
\end{subfigure}%
\begin{subfigure}{0.25\linewidth}
\centering\includegraphics[width=1.1\linewidth]{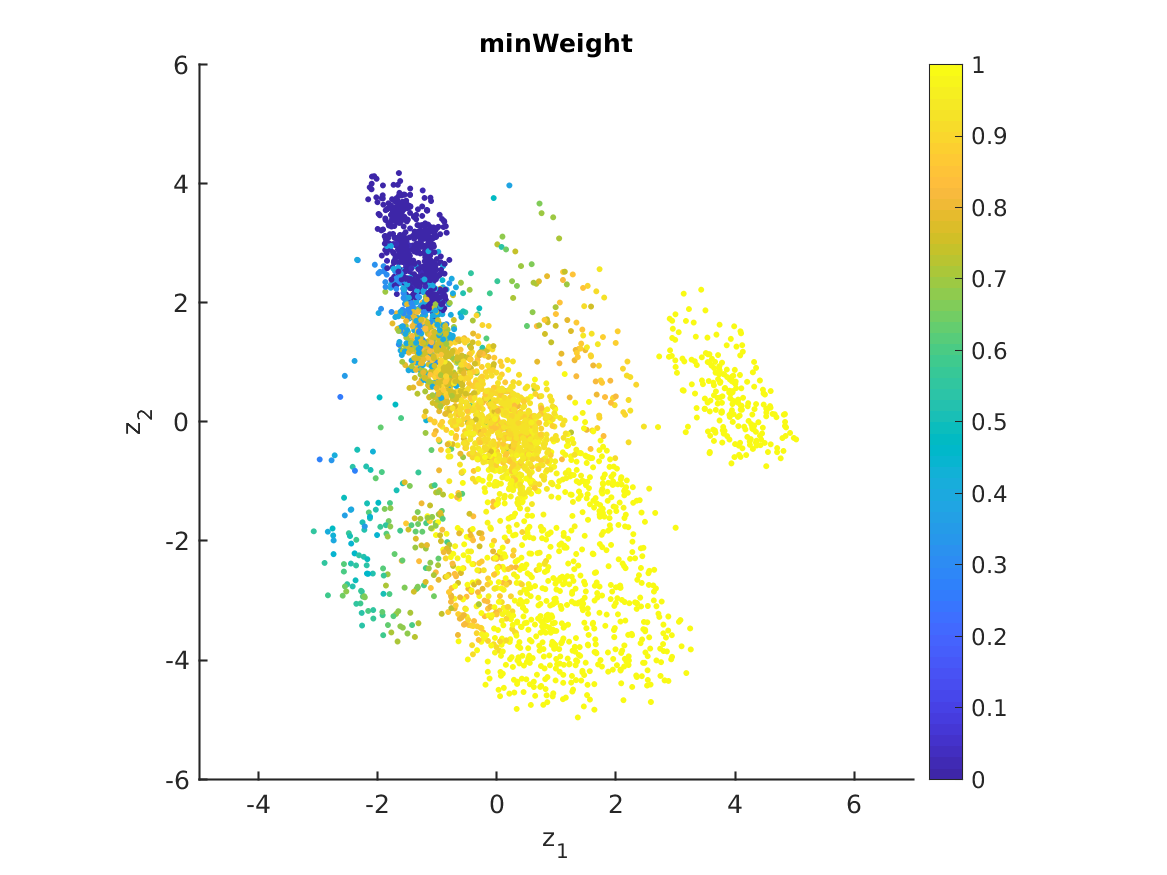}
\caption{}
\label{minWeightSubfigure}
\end{subfigure}
 \begin{subfigure}{0.25\linewidth}
\centering\includegraphics[width=1.1\linewidth]{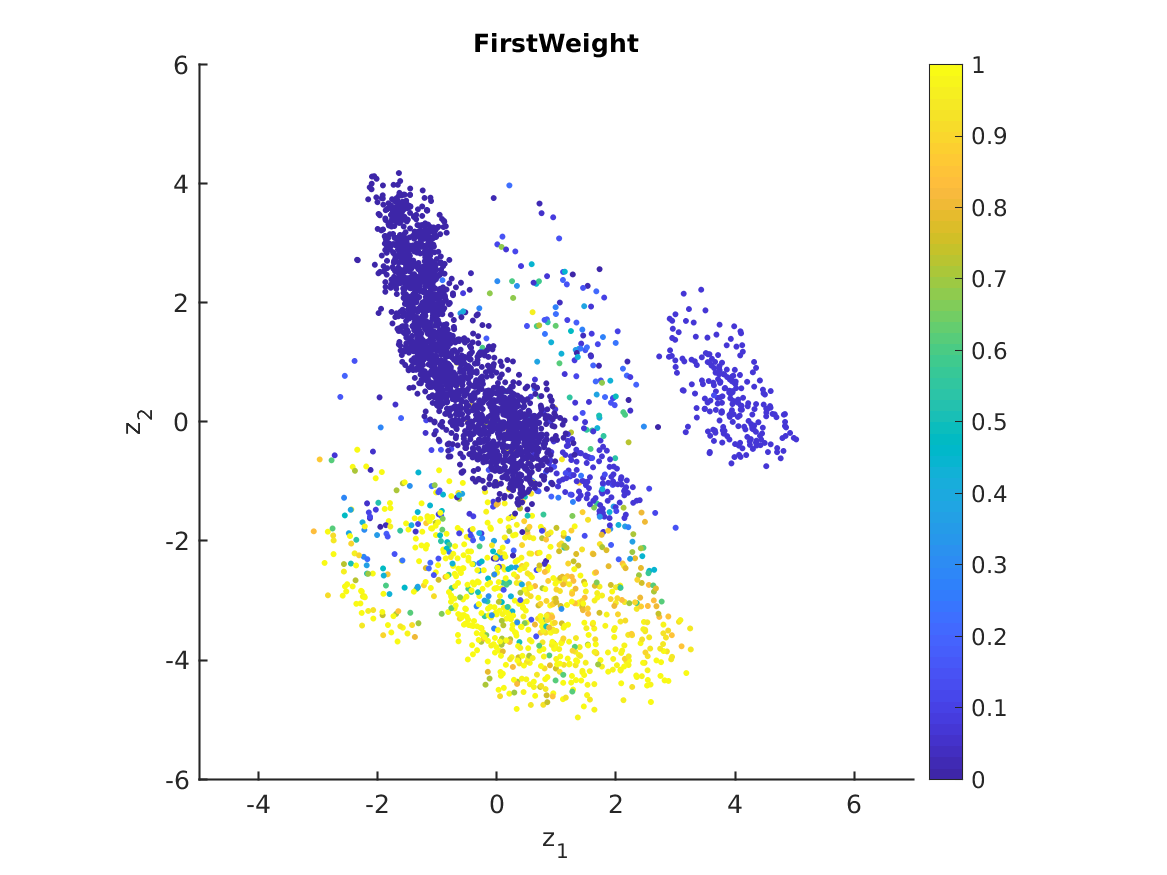}
\caption{}
\label{firstWeightSubfigure}
\end{subfigure}\vspace{10pt}

\begin{subfigure}{0.25\linewidth}
\centering\includegraphics[width=1.1\linewidth]{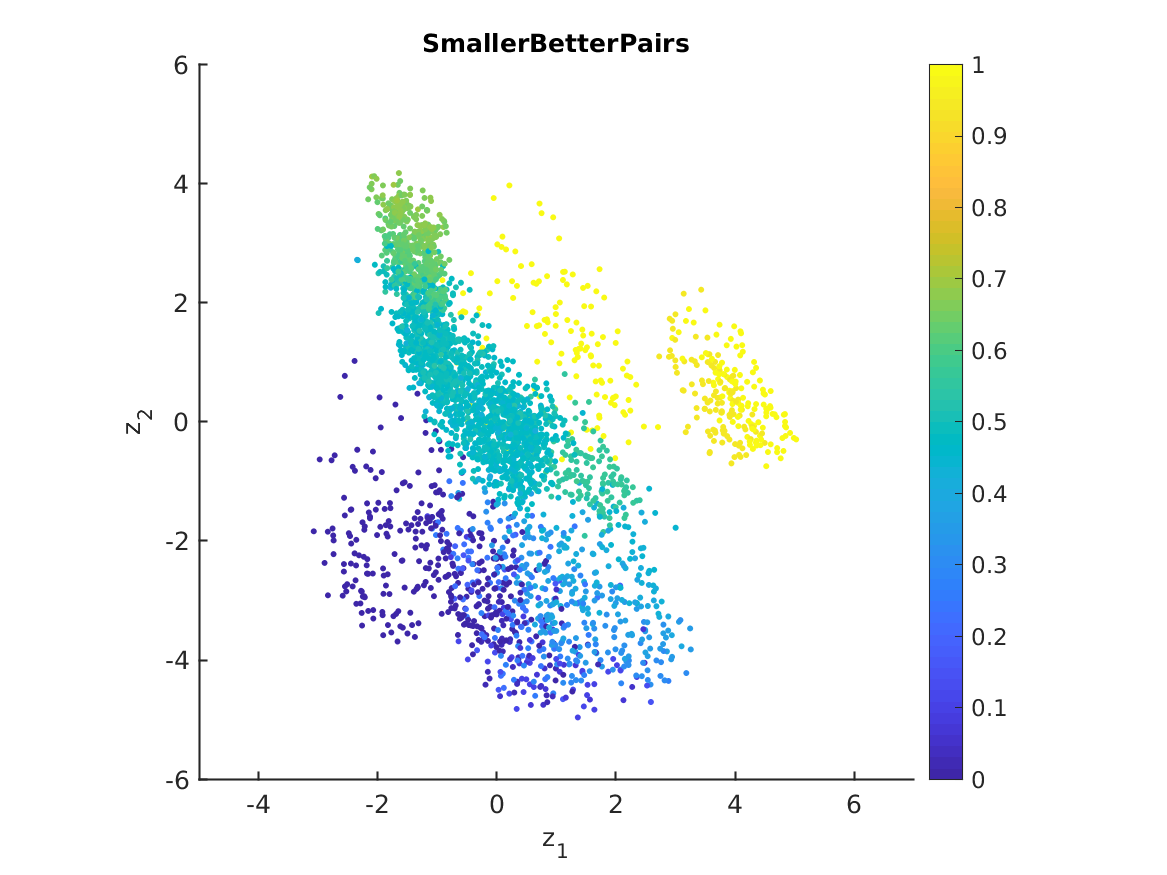}
\caption{}
\label{smallerBetterPairsSubfigure}
\end{subfigure}%
\begin{subfigure}{0.25\linewidth}
\centering\includegraphics[width=1.1\linewidth]{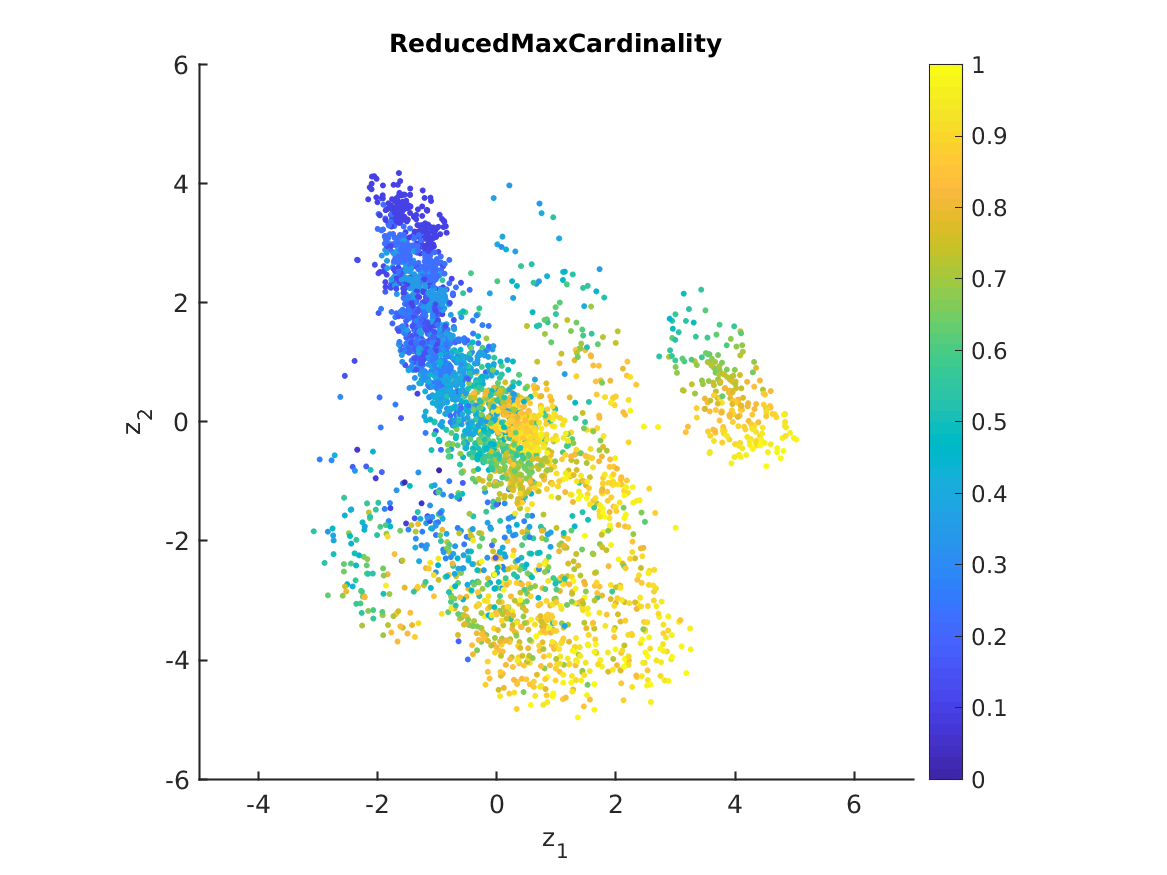}
\caption{}
\label{reducedMaximumCardinalitySubfigure}
\end{subfigure}
 \begin{subfigure}{0.25\linewidth}
\centering\includegraphics[width=1.1\linewidth]{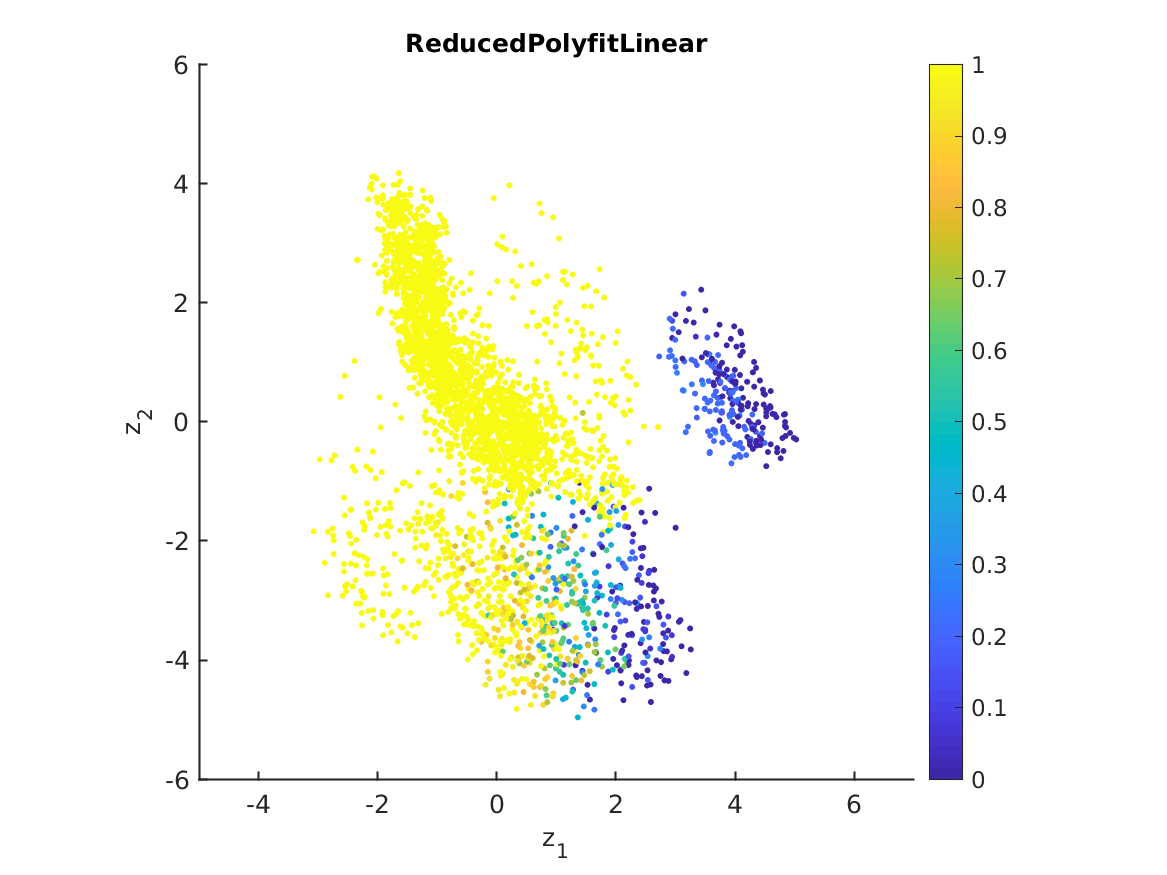}
\caption{}
\label{reducedPolyfitLinearSubfigure}
\end{subfigure}\vspace{10pt}
\caption{The 9 normalized features visualized in a 2D problem instance space.}
\label{9FeaturesFigure}
\end{figure}

	\section{Conclusions and further work}
\label{conclusionsSection}
In this paper, we formulated several computationally challenging problems related to the IMSs of a 0-1 knapsack problem instance. We proved several structural results of IMSs and based on this we formulated the first polynomial and pseudopolynomial time algorithms for solving these computationally challenging problems. We generalized two theorems from earlier work \citep{Jooken:2022} on noisy multi-group exponential 0-1 knapsack problem instances to arbitrary 0-1 knapsack problem instances. These theorems demonstrate a lower bound for the objective function value of an IMS and a cardinality constraint on the group consisting of the least heavy items. A set of 14 features was derived from the results of these computationally challenging problems and from subexpressions that play a crucial role in the theorems. All algorithms were implemented and we calculated these 14 features, together with 44 features from the literature, for two large datasets of problem instances by using a supercomputer for approximately 540 CPU-hours. Several experiments were performed in which these features were used by machine learning models to obtain insights into the empirical hardness of a problem instance by considering the running time of the state-of-the-art knapsack algorithm \textit{\textbf{Combo}} . These experiments show that the proposed features contain important information related to the hardness of a problem instance that was not present in earlier features from the literature and the proposed features are complementary with the features from the literature for one out of two datasets. Furthermore, the experiments indicate that the properties of IMSs are important hardness indicators for a wide variety of 0-1 knapsack problem instances. The instance space analysis methodology was used to visualize the problem instances in a two-dimensional space and this revealed that the locations of hard 0-1 knapsack problem instances are clustered together around a relatively dense region of the instance space. By visualizing the features in the same space, we were able to observe several patterns that were often different in the easy and hard parts of the instance space, which makes these features useful to distinguish easy and hard problem instances.

There are several opportunities for future work that deserve further attention. The algorithms that we propose for calculating the features are all exact algorithms that work in polynomial or pseudopolynomial time, which made it possible for the first time to calculate these features for problem instances with a relatively large number of items. However, the knapsack capacity is often a limiting factor and these algorithms cannot handle very large knapsack capacities well. It is interesting to investigate whether alternative algorithms can be found that scale better for large knapsack capacities. Additionally, it would be interesting to investigate what the impact would be of considering heuristic algorithms that approximate the features to obtain a reduced computational cost. For some features (e.g. the clustering related feature $g^*$) such heuristic algorithms are readily available, whereas for other features (e.g. $|X|$) it is not clear how such an algorithm could be obtained and we think it is likely that more structural insights would be necessary.

It is also interesting to investigate whether the theorems that were developed in this paper can be integrated into existing knapsack solvers to obtain faster algorithms. In particular, the lower bound constraints and cardinality constraints that were proven for IMSs seem to be good candidates. Since the optimal solution of a 0-1 knapsack problem instance is an IMS, such constraints can be added to the integer programming formulation without changing the optimum and this new problem can be relaxed to obtain a dual bound (note that this was also one of the key ingredients of the \textit{\textbf{Combo}} algorithm). It is worthwhile investigating which sort of relaxations would yield the tightest bounds and how these relaxations could be efficiently solved.

\section*{Acknowledgments}

The computational resources and services used in this work were provided by the VSC (Flemish Supercomputer Center), funded by the Research Foundation - Flanders (FWO) and the Flemish Government - department EWI. We gratefully acknowledge the support provided by the ORDinL project (FWO-SBO S007318N, Data Driven Logistics, 1/1/2018 - 31/12/2021). This research received funding from the Flemish Government under the ``Onderzoeksprogramma Artifici\"{e}le Intelligentie (AI) Vlaanderen'' programme.

\bibliography{references}

\newpage
\newpage
\pagebreak
\newpage
\section*{Appendix}
\renewcommand{\thefigure}{A\arabic{figure}}
\begin{figure}[!ht]
\centering\includegraphics[width=0.8\linewidth]{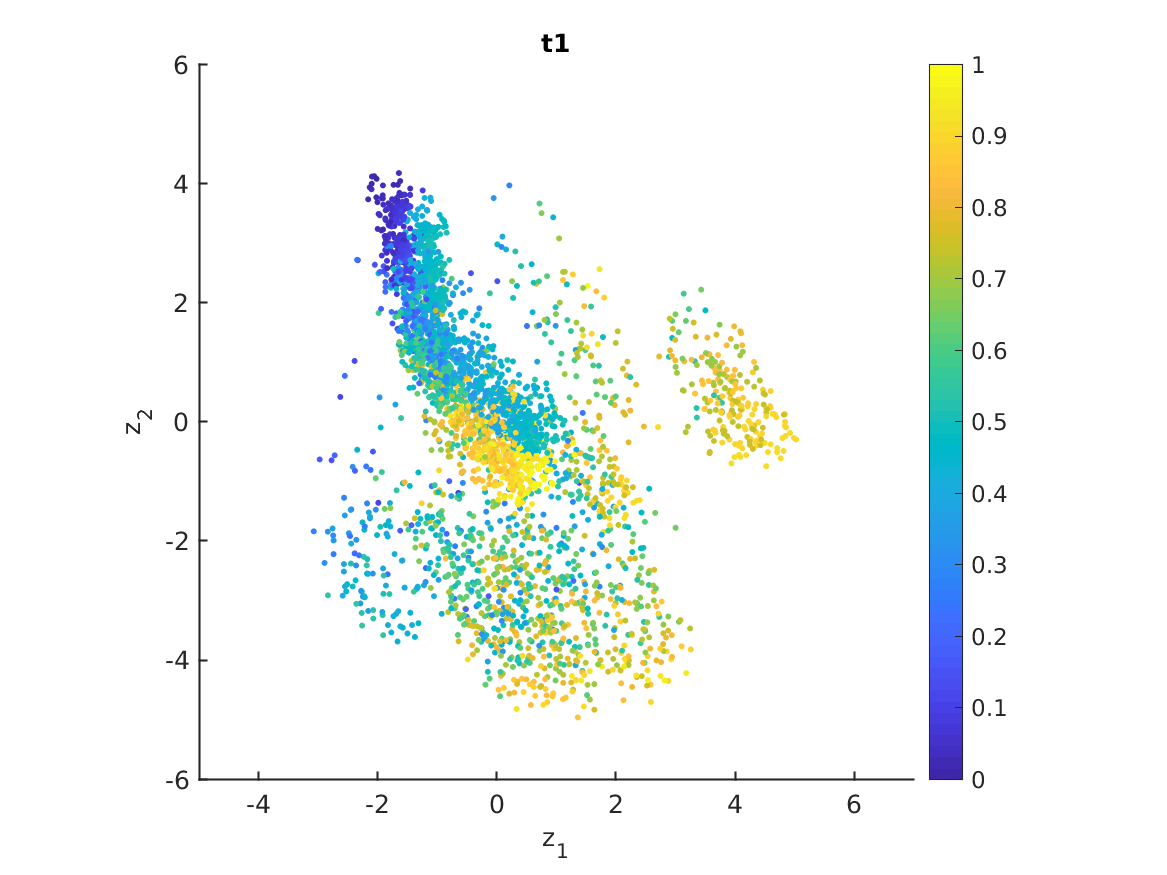}
\caption{The feature $t_1$ visualized in a 2D problem instance space.}
\label{t1Figure}
\end{figure}

\begin{figure}[!ht]
\centering\includegraphics[width=0.8\linewidth]{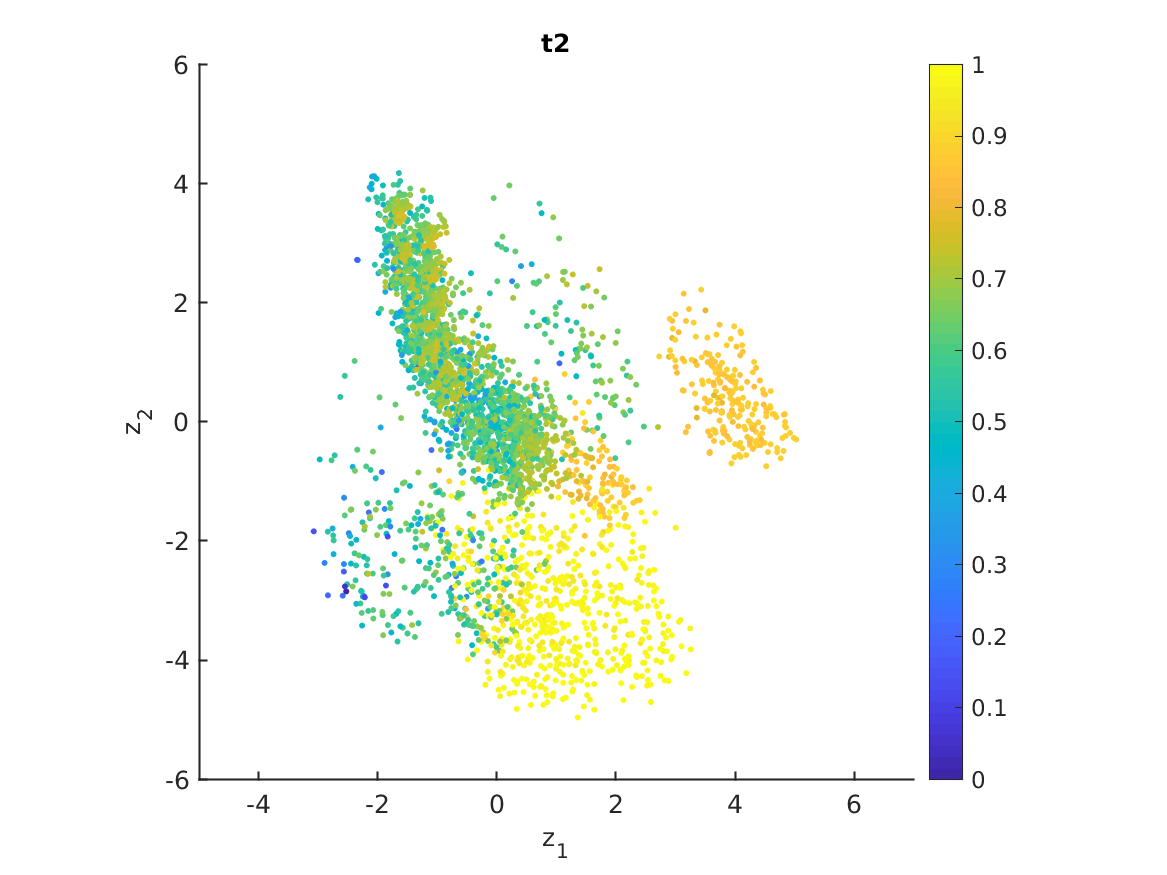}
\caption{The feature $t_2$ visualized in a 2D problem instance space.}
\end{figure}

\begin{figure}[!ht]
\centering\includegraphics[width=0.8\linewidth]{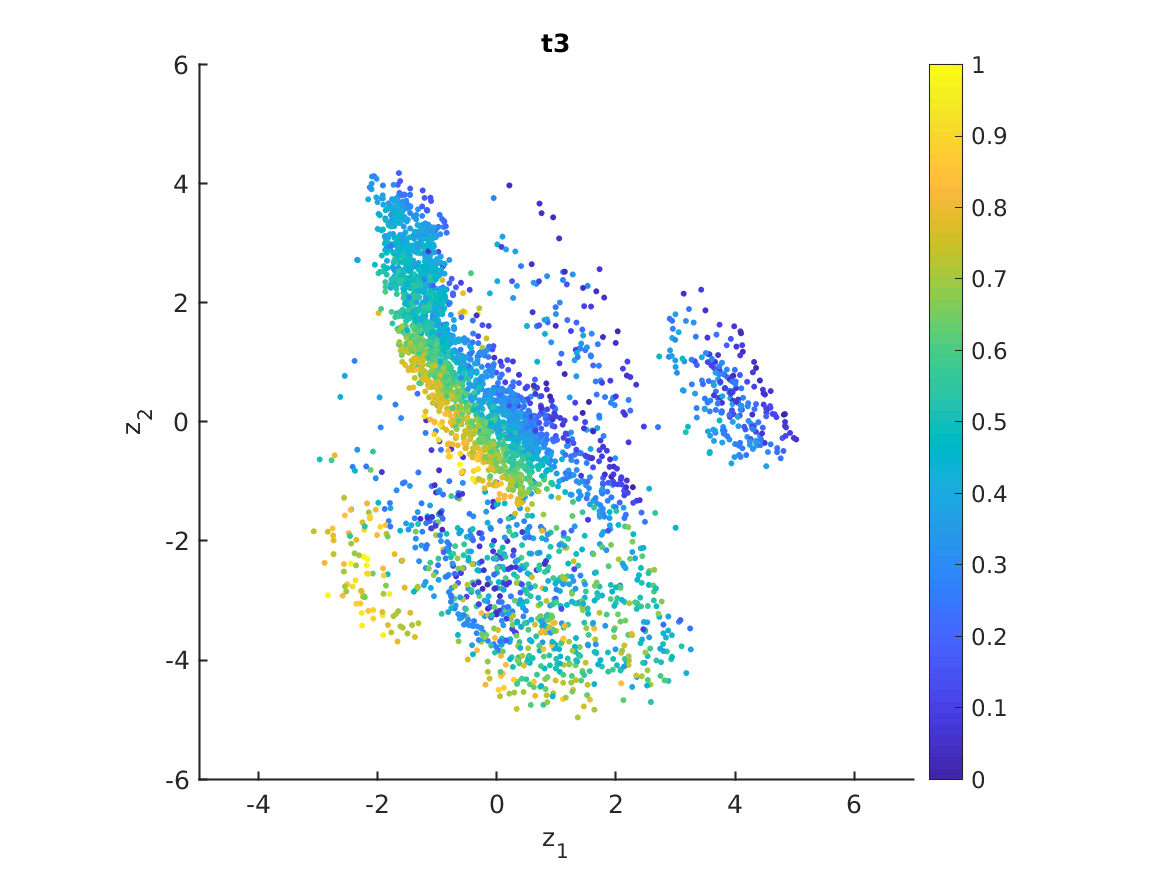}
\caption{The feature $t_3$ visualized in a 2D problem instance space.}
\end{figure}

\begin{figure}[!ht]
\centering\includegraphics[width=0.8\linewidth]{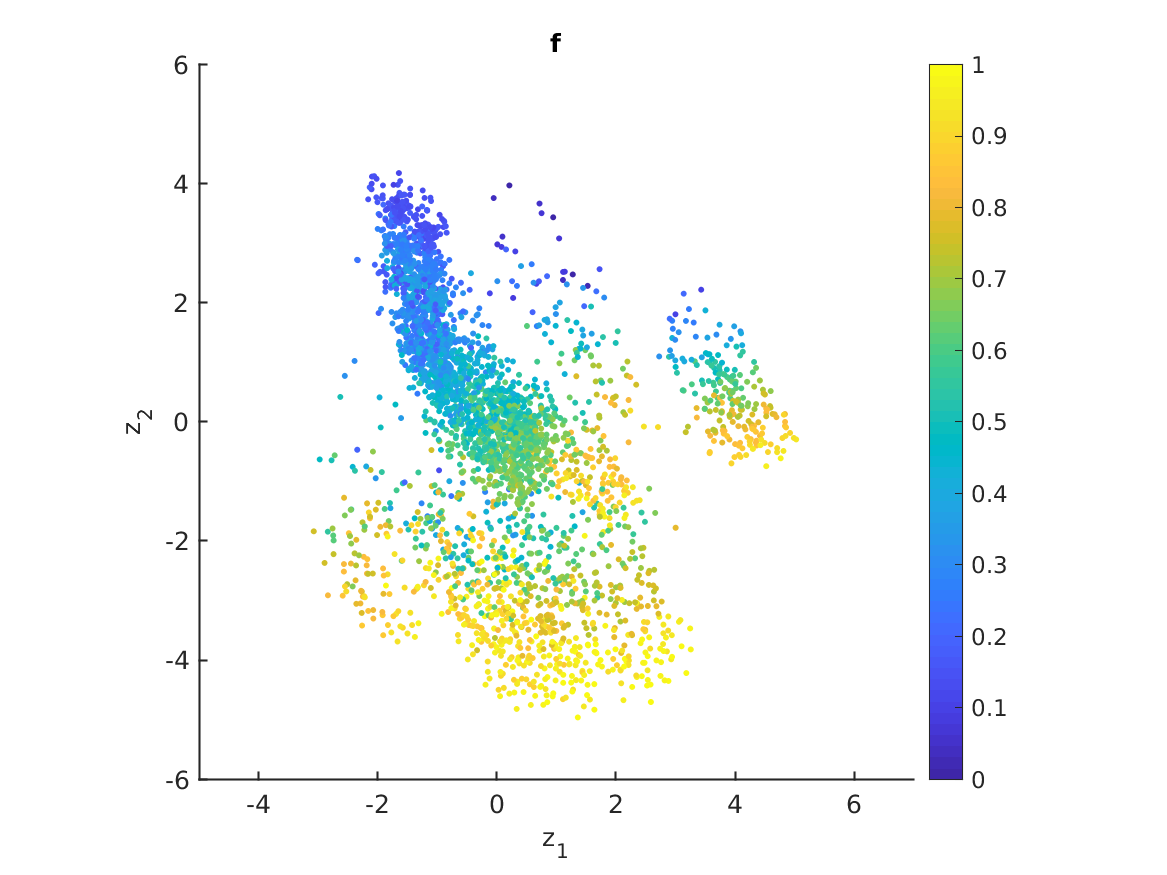}
\caption{The feature $f$ visualized in a 2D problem instance space.}
\end{figure}

\begin{figure}[!ht]
\centering\includegraphics[width=0.8\linewidth]{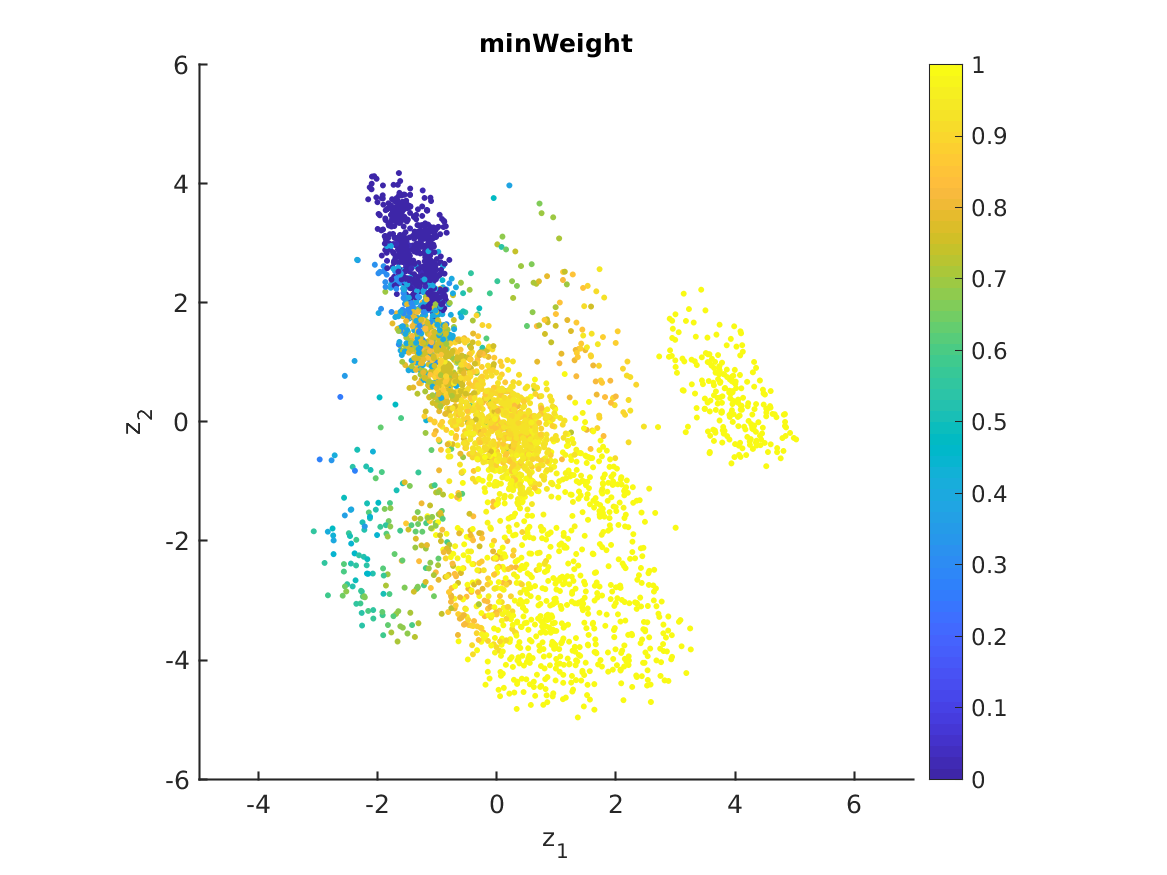}
\caption{The feature Minimum Weight visualized in a 2D problem instance space.}
\end{figure}

\begin{figure}[!ht]
\centering\includegraphics[width=0.8\linewidth]{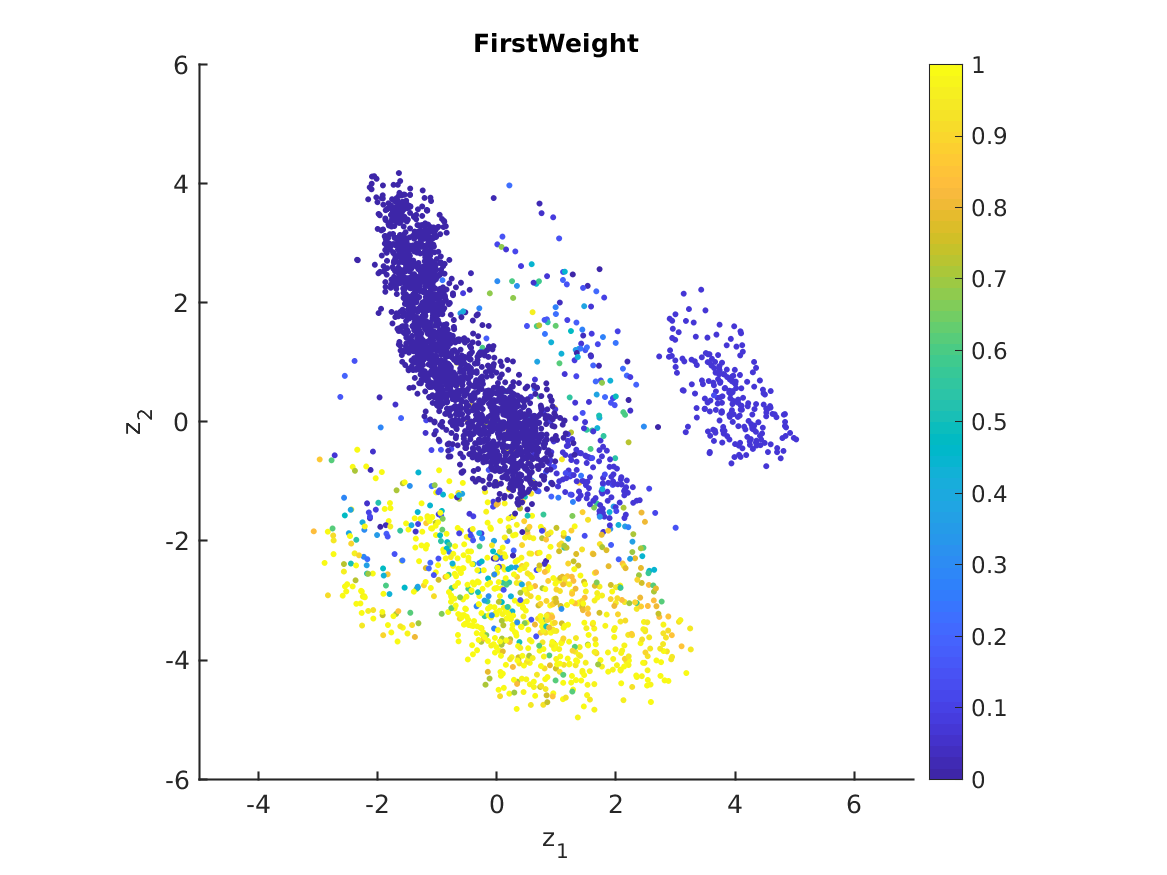}
\caption{The feature First Weight visualized in a 2D problem instance space.}
\end{figure}

\begin{figure}[!ht]
\centering\includegraphics[width=0.8\linewidth]{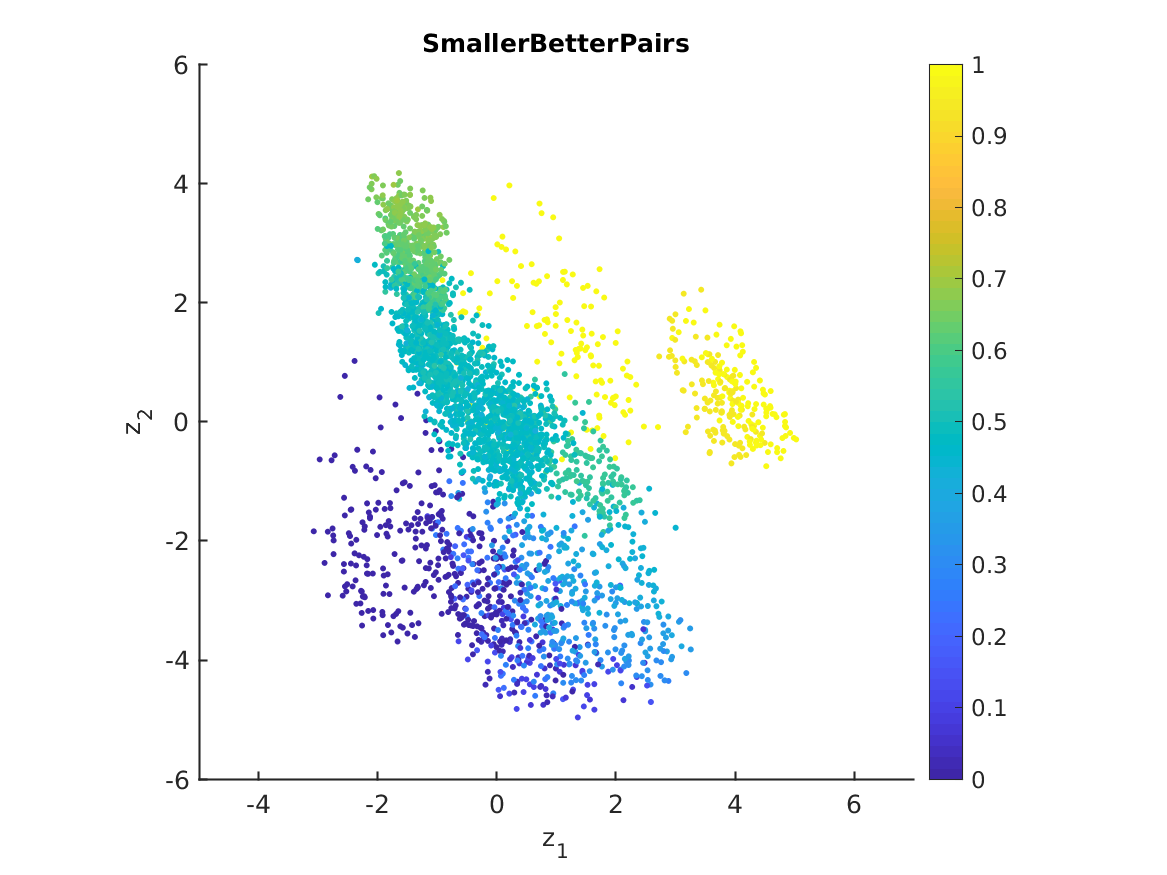}
\caption{The feature Smaller Better Pairs visualized in a 2D problem instance space.}
\end{figure}

\begin{figure}[!ht]
\centering\includegraphics[width=0.8\linewidth]{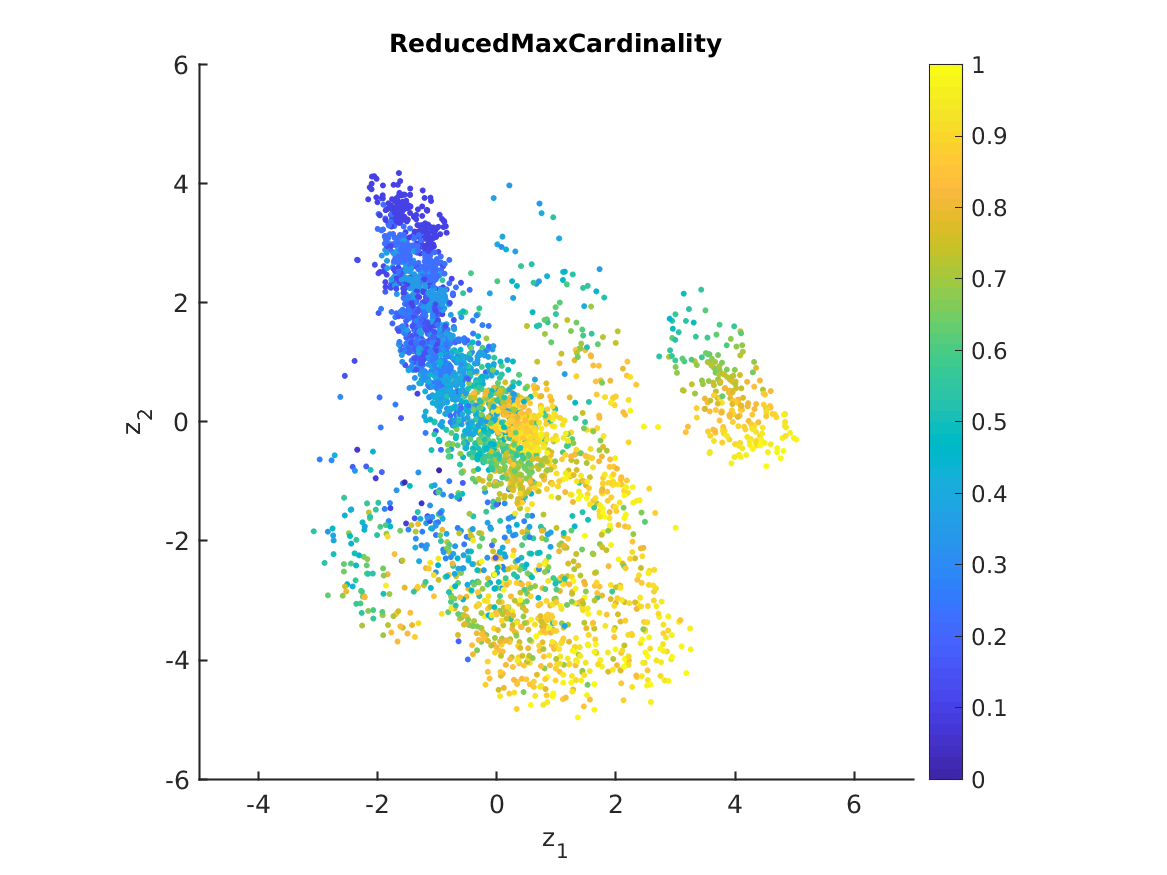}
\caption{The feature Reduced Maximum Cardinality visualized in a 2D problem instance space.}
\end{figure}

\begin{figure}[!ht]
\centering\includegraphics[width=0.8\linewidth]{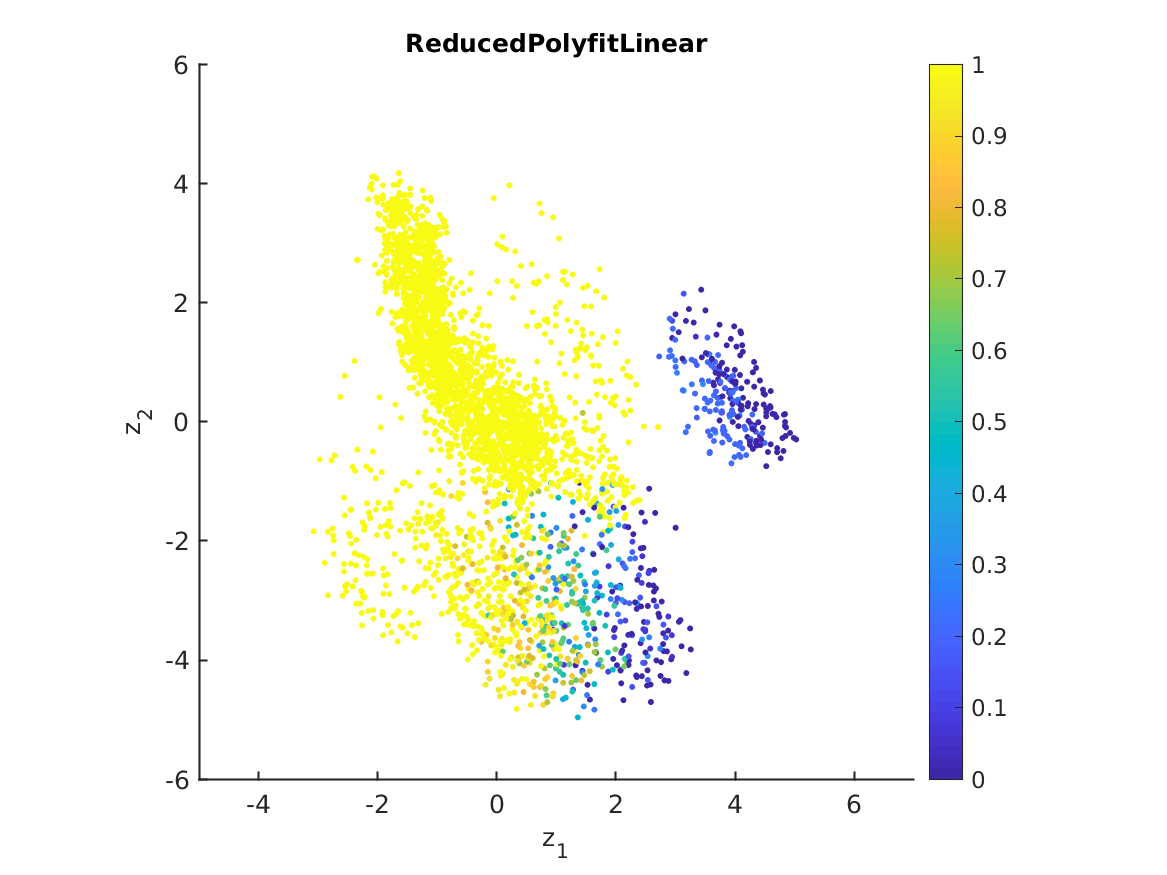}
\caption{The feature Reduced Polyfit Linear visualized in a 2D problem instance space.}
\label{reducedPolyfitLinearFigure}
\end{figure}

\end{document}